\pdfoutput=1 
\documentclass[11pt, a4paper]{article}
\usepackage{bcc23}
\usepackage{epsfig}

\newif\ifcolourfigs
\colourfigstrue 

\newcommand{\F}{{\ensuremath {\mathcal F}}}
\newcommand{\G}{{\ensuremath {\mathcal G}}}
\newcommand{\GG}{{\ensuremath {\mathcal G}{\mathcal G}}}
\newcommand{\numplayers}{k}
\newcommand{\numstrats}{n}
\newcommand{\pindex}{i}
\newcommand{\aindex}{j}
\newcommand{\msprob}{x}
\newcommand{\cube}{K}

\newcommand{\Poly}{{\bf P}}
\newcommand{\NP}{{\bf NP}}
\newcommand{\PPAD}{{\bf PPAD}}
\newcommand{\PPA}{{\bf PPA}}
\newcommand{\PPP}{{\bf PPP}}
\newcommand{\PSPACE}{{\bf PSPACE}}
\newcommand{\hide}[1]{}

\bcctitle{A Survey of PPAD-Completeness for Computing Nash Equilibria\footnote{This paper
accompanies an invited talk at the 23rd British Combinatorial Conference (2011). It
appears as a chapter in {\em Surveys in Combinatorics} 2011 (Cambridge University Press)}}
\bccshorttitle{Games and PPAD-completeness}

\bccname{Paul W. Goldberg}
\bccshortname{P.W. Goldberg}

\bccaddressa{Dept. of Computer Science, Ashton Building, Ashton Street, Liverpool L69 3BX}
\bccemaila{P.W.Goldberg@liverpool.ac.uk}

\begin{document}

\def\appxnash{{\epsilon}}
\def\pr{\text{Pr}}
\newlength{\boxwidth}
\setlength{\boxwidth}{\linewidth}
\addtolength{\boxwidth}{-1em}

\makebcctitle

\begin{abstract}
\PPAD\ refers to a class of computational problems for which solutions
are guaranteed to exist due to a specific combinatorial principle. The
most well-known such problem is that of computing a Nash equilibrium of
a game. Other examples include the search for market equilibria, and
envy-free allocations in the context of cake-cutting. A problem is said
to be complete for \PPAD\ if it belongs to \PPAD\ and can be shown to
constitute one of the hardest computational challenges within that class.

In this paper, I give a relatively informal overview of the proofs used
in the \PPAD-completeness results. The focus is on the mixed Nash
equilibria guaranteed to exist by Nash's theorem. I also give an overview
of some recent work that uses these ideas to show
\PSPACE-completeness for the computation of specific equilibria found
by homotopy methods. I give a brief introduction to related problems of
searching for market equilibria.
\end{abstract}
\thankyou{Currently supported by EPSRC Grant EP/G069239/1}%


\hide{
\section{Polynomial-time computation}

One of the great achievements of theoretical Computer Science was the
identification of ``polynomial-time computation'' as a mathematically
precise notion of an efficient algorithm. In assessing the performance
of an algorithm, we consider the way its runtime scales as a function
of the size of an input. If $n$ measures the size of an input, an algorithm
runs in polynomial time provided that its runtimes is proportional to
$n$, or $n^2$, or $n^3$, or generally at most $n^k$ for some constant $k$.
On the other hand, a runtime proportional to (say) $2^n$ is not polynomial-time.

The great thing about the definition is that it is robust to changes in the
model of computation --- in moving from (say) Turing machine to unlimited
register machine, the degree of the polynomial may change, but we would
still have a polynomial bound on the runtime.
}

\section{Total Search Problems}

Suppose that you enter a maze without knowing anything in advance about
its internal structure. Let us assume that it has only one entrance. To
solve the maze, we do not ask to find some central chamber whose
existence has been promised by the designer. Instead you need to find
either a dead end, or else a place where the path splits, giving you a
choice of which way to go; see Figure~\ref{fig:mazes}. Observe that this
kind of solution is guaranteed to exist, and does not require any kind
of promise. This is because, if there are no places where the explorer
has a choice, then the interior of the maze (at least, the parts accessible from
the entrance) is, topologically, a single
path leading to a dead end. If you are asked to find either a dead end
or a split of a path, this is informally an example of a (syntactic)
{\em total search problem} --- the problem description has been set up
so as to guarantee that there exists a solution; the word ``total''
refers to the fact that {\em every} problem instance has a solution.

A maze of the sort in Figure~\ref{fig:mazes} could of course be solved
in linear time by checking each location. Consider now a more
challenging problem defined as follows:
\begin{Def}
{\sc Circuit maze} is a search problem on a $2^n\times 2^n$ grid --- the
maze is specified using a boolean circuit $C$ that takes as input a bit
string of length $O(n)$ that represents the location (coordinates) of a
possible wall or barrier between 2 adjacent grid points. $C$ has a
single output bit that indicates whether in fact a barrier exists at
that location.
\end{Def}
A naive search for a grid point that corresponds to a dead end or a
split of a path is no longer feasible, so the search problem becomes
nontrivial. Notice that not only must some solution exist, but in
addition it is easy to use $C$ to check the validity of a claimed
solution.

The problem of finding, or computing, a Nash equilibrium of a game
(defined in Section~\ref{sec:ne}) is analogous. Nash's
theorem~\cite{Nash} assures us in advance that every game has a Nash
equilibrium, and ultimately it works by applying a similar combinatorial
principle (explained below) to the one that is being used to assure
ourselves that maze
problems of the sort defined above must have a solution.

\begin{figure}
\begin{center}
\ifcolourfigs
\includegraphics[width=5.0in]{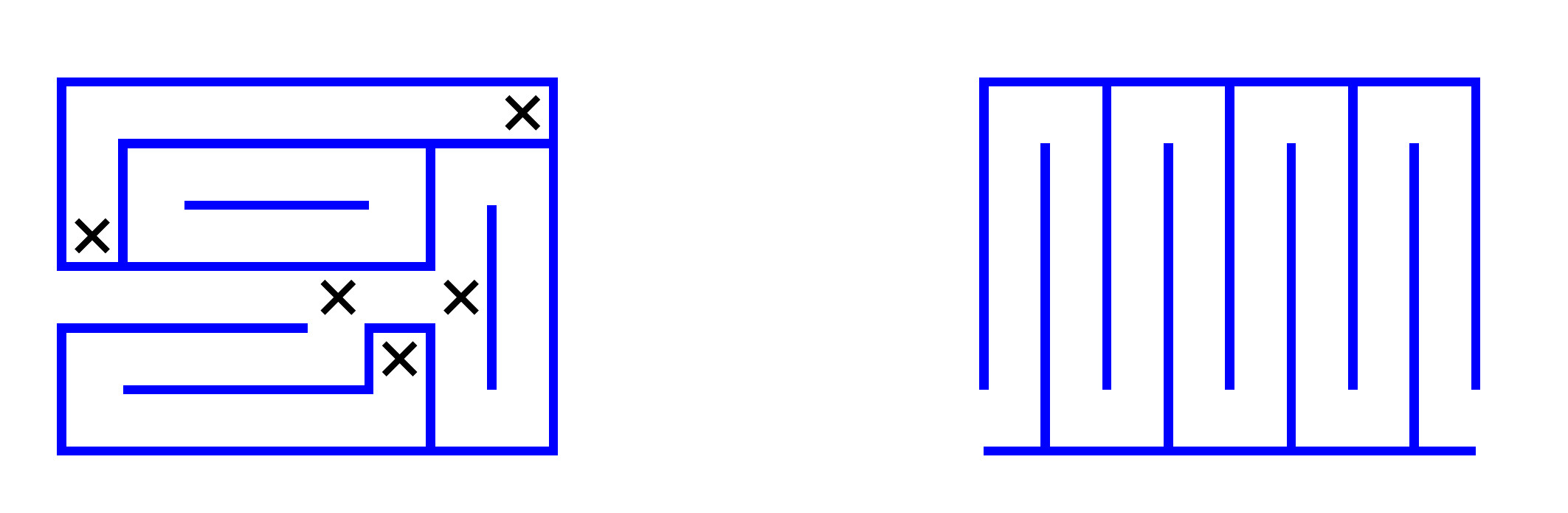}
\else
\includegraphics[width=5.0in]{mazes.pdf}
\fi
\end{center}
\caption{The left-hand maze (with a single entrance on the left) has 5
solutions marked with crosses. Note that while solutions need not be
accessible starting from the entrance, at least one will be accessible.
The right-hand maze illustrates how there need not be any solutions when
there are 2 entrances. In general, an odd number of entrances guarantees
the existence of a solution.
}\label{fig:mazes}
\end{figure}

\subsection{\NP\ Total Search Problems}

Why can we not apply a more standard notion of computational hardness,
such as \NP-hardness, to the problem of Nash equilibrium computation? It
turns out to be due precisely to its special status as a {\em total}\,
search problem, as informally defined above.

Standard \NP-hard problems do indeed have associated optimisation
problems that are total search problems, but they are not \NP\ total
search problems. Consider for example the (\NP-complete) travelling
salesman problem, commonly denoted {\sc TSP}. This has an associated
optimisation problem, in which we seek a tour of minimal length that
visits all the cities. By definition, such a minimal-length tour must
exist, so this total search problem is \NP-hard. But it is not, as far
as we know, a member of \NP\ --- given a solution there is no obvious
efficient way to check its optimality.

The complexity class \PPAD\ (along with related classes introduced by
Papadimitriou~\cite{Pap}) was intended to capture the computational
complexity of a relatively small number of problems that seem not to
have polynomial-time algorithms, but where there is a mathematical
guarantee that every instance has a solution, and furthermore, {\em
given a solution, the validity of that solution may be checked in
polynomial time}. Each class of problems (see
Section~\ref{sec:PPAD-related}) has an associated combinatorial
principle that guarantees that one has a total search problem.

\medskip
A simple result due to Megiddo~\cite{Meg} shows that if such a problem
is \NP-complete, then \NP\ would have to be equal to co-\NP --- it is
proved as follows.

Suppose we have a reduction from any \NP-hard problem (e.g. {\sc sat})
to any \NP\ total search problem (e.g. {\sc Nash}). Thus, from any {\sc
sat}-instance (a propositional formula) we can efficiently construct
a {\sc Nash}-instance (a game) so that
given any solution (Nash equilibrium) to that {\sc Nash}-instance we can (efficiently)
derive an answer to the {\sc sat}-instance. That reduction could then be used
to construct a nondeterministic algorithm for verifying that an
unsatisfiable instance of {\sc sat} indeed has {\em no} solution:  Just
guess a solution of the {\sc Nash}-instance, and check that it indeed
fails to identify a solution for the {\sc sat}-instance.

The existence of such a nondeterministic algorithm for {\sc sat} (one
that can verify that an unsatisfiable formula is indeed unsatisfiable,
hence implying that \NP=co-\NP) is an eventuality that is considered by
complexity theorists almost as unlikely as \Poly=\NP. We conclude that
{\sc Nash} is very unlikely to be \NP-complete.

\subsection{Reducibility among total search problems}\label{sec:reductions}

Suppose we have two total search problems $X$ and $Y$. We say that $X$
is reducible to $Y$ in polynomial time if the following holds. There
should be functions $f$ and $g$ both computable in polynomial time,
such that given an instance $I_X$ of $X$, $f(I_X)$ is an instance of
$Y$, and given any solution $S$ to $f(I_X)$, $g(S)$ is a solution to
$I_X$. Thus, if we had a polynomial-time algorithm that solves $Y$, the
reduction would construct a polynomial-time algorithm for $X$. So, such
a reduction shows that $Y$ is ``at least as hard'' as $X$.

While reductions in the literature are of the above form, one could use
a less restrictive definition, called a Turing reduction, in which
problem $X$ reduces to problem $Y$ provided that we we can write down an
algorithm that solves $X$ in polynomial time, provided that it has
access to an ``oracle'' for problem $Y$. As a consequence, if $Y$ does
in fact have a polynomial-time algorithm then so does $X$. However, the
reductions used in the literature to date about total search problems,
are of the more restricted type.

\subsection{\PPAD, and some related concepts}\label{sec:PPAD-related}

\PPAD, introduced in~\cite{Pap}, stands for ``polynomial parity argument
on a directed graph''. It is defined in terms of a rather
artificial-looking problem {\sc End of the line}, which is the
following:\footnote{This is called {\sc End-of-Line} in~\cite{CDT},
which is arguably a better name since the ``line'' is not unique, and
there is no requirement that we find the end of any specific line.}

\begin{Def} An instance of {\sc End of the line} consists of two boolean
circuits $S$ and $P$ each of which has $n$ inputs and $n$ outputs, such that
$P(0^n)=0^n\not=S(0^n)$. Find a bit vector $x$ such that $P(S(x))\not=x$
or $S(P(x))\not=x \not=0^n$. \footnote{Chen et al.~\cite{CDT} define
it in terms of a single circuit that essentially combines $S$ and $P$,
that takes an $n$-bit vector $v$ as input, and outputs $2n$ bits that
correspond to $S(v)$ and $P(v)$.}
\end{Def}

$S$ and $P$ (standing for {\em successor} and {\em predecessor})
implicitly define a digraph $G$ on $2^n$ vertices (bit strings of length
$n$) in which each vertex has indegree and outdegree at most 1. $(v,w)$
is an arc of $G$ (directed from $v$ to $w$) if and only if $S(v)=w$ and $P(w)=v$. By construction,
$0^n$ has indegree 0, and either has outdegree 1, or $P(S(0^n))\not=0^n$ (in
which case $0^n$ is a solution). Notice that $G$ permits efficient
local exploration (the neighbours of any vertex $v$ are easy to compute
from $v$) but non-local properties are opaque. We shall refer to a graph
$G$ that is represented in this way as an $(S,P)$-graph.

The ``parity argument on a directed graph'' refers to a more general observation:
define an ``odd'' vertex of a graph to be one where the total number of incident edges
is an odd number. Then notice that the number of odd vertices must be an even number.
Indeed, this observation applies generally to undirected graphs, but we apply it here
to $(S,P)$-graphs, in particular those $(S,P)$-graphs where vertex
$0^n$ actually has an outgoing arc (so has odd degree).
These $(S,P)$-graphs have an associated total search problem, of finding
an alternative odd-degree vertex. One could search for such a vertex by, for example
\begin{itemize}
\item checking each bit string in order, looking up its neighbours to see if it's an
odd vertex, or
\item following a directed path starting at some vertex; an endpoint
other than $0^n$ must be reached,
\end{itemize}
but since $G$ is exponentially large, these naive approaches will take
exponential time in the worst case.

\begin{Def} A computational problem $X$ belong to the complexity class \PPAD\
provided that $X$ reduces to {\sc End of the line} in polynomial time. Problem $X$
is \PPAD-complete provided that $X$ is in \PPAD, and in addition {\sc End of the line}
reduces to $X$ in polynomial time.
\end{Def}

Thus {\sc End of the line} stands in the same relationship to \PPAD,
that {\sc Circuit sat} does to \NP\ (although \NP\ is not actually defined in
terms of {\sc Circuit sat}, an equivalent definition would say that \NP-complete
problems are those that are polynomial-time equivalent to {\sc Circuit sat},
and a member of \NP\ is a problem that can be reduced to {\sc Circuit sat}).

The above definition of \PPAD\ is the one used in~\cite{DGP}; it is
noted there that there are many alternative equivalent
definitions\footnote{The original definition of~\cite{Pap} is in terms
of a Turing machine rather than circuits.}. Since {\sc End of the line}
is by construction a total search problem, it follows that members of
\PPAD\ are necessarily also total search problems.

\begin{paragraph}{How hard is ``\PPAD-complete''?}
\PPAD\ lies ``between \Poly\ and \NP'' in the sense that if \Poly\ were equal to \NP, then
all \PPAD\ problems would be polynomial-time solvable, while the assumption
that ``\PPAD-complete problems are hard'' implies \Poly\ not equal to \NP. It is
a fair criticism of these results that they do not carry as much weight as
do \NP-hardness results, partly for this reason, and partly because there
are only a handful of \PPAD-complete problems, while thousands of problems
have been shown to be \NP-complete.

Why, then, do we take \PPAD-completeness as evidence that a problem cannot be
solved in polynomial time? One argument is that {\sc End of the line} is
defined in terms of unrestricted circuits, and general boolean circuits seem
to be hard to analyse via polynomial-time algorithms. We should also note the oracle
separation results of~\cite{BCEIP} in this context.
\end{paragraph}

\begin{paragraph}{Related complexity classes}
Other combinatorial principles are considered in~\cite{Pap},
that guarantee totality of corresponding search problems.
For example, consider the pigeonhole principle, that given a function $f:X\longrightarrow Y$,
if $X$ and $Y$ are finite and $|X|>|Y|$, there must exist $x,x'\in X$ such that
$f(x)=f(x')$. Now define an associated computational problem:
\begin{Def}
The problem {\sc Pigeonhole circuit} has as instances, directed boolean circuits having
the same number of inputs and outputs. Let $f$ be the function computed by the circuit.
The problem is to identify {\em either} two distinct bit strings $x$, $x'$ with $f(x)=f(x')$,
{\em or} a bit string $x$ with $f(x)=0$ (where $0$ is the all-zeroes bit string).
\end{Def}
Note that by construction, this a total search problem, and it is an \NP\ total search problem
since it is computationally easy to check that a given solution is valid. At the same
time, it {\em seems} to be hard to find a solution, although it is
unlikely to be \NP-hard, due to Megiddo's result. The complexity class \PPP~\cite{Pap}
(for ``polynomial pigeonhole principle'') is defined as the set of all total search problems
that are reducible to {\sc Pigeonhole circuit}.

The pigeonhole principle is a generalisation of the parity argument on a directed graph.
To see this, notice that the function $f:X\longrightarrow Y$ can map a vertex $v$ of an
{\sc End of the line} graph to the adjacent vertex connected by an arc from $v$, or
to itself if $v$ has outdegree 0.
\PPAD\ is a subclass of \PPP --- {\sc End of the line} reduces to {\sc Pigeonhole circuit};
it would be nice to obtain a reduction the other way and show equivalence, but that
has not been achieved.

If we have an {\em undirected} graph of degree at most 2 with a known endpoint, then the search
for another endpoint is also a total search problem. The corresponding complexity class
defined in~\cite{Pap} is \PPA. The {\sc Circuit maze} problem that we considered informally
at the start, belongs to \PPA. As it happens, the problem is likely to be complete for \PPA.
\end{paragraph}

\section{Sperner's lemma, and an associated computational problem}

Sperner's lemma is the following combinatorial result, that can be used to prove Brouwer's
fixed point theorem.
\begin{theorem}\cite{Sperner}
Let $\{v_0,v_1,\ldots,v_d\}$ be the vertices of a $d$-simplex $S$, and suppose that
the interior of $S$ is decomposed into smaller simplices using additional vertices.
Assign each vertex a colour from $\{0,1,\ldots d\}$ such that $v_i$ gets colour $i$, and a vertex on
any face of $S$ must get one of the colours of the vertices of that face.
Interior vertices may be coloured arbitrarily.
Then, this simplicial decomposition must include a
{\em panchromatic simplex}, i.e. one whose vertices has all distinct colours.
\end{theorem}

\begin{proof}
The proof can be found in many places, so here we just give a sketch for the 2-dimensional
case. Let us define the computational problem {\sc Sperner} (discussed in more detail below)
to be the problem of exhibiting a trichromatic triangle.
Essentially, the proof consists of a reduction from {\sc Sperner} to {\sc End of the line}!

Choose any two of the colours, say 0 and 1. We begin by adding some further triangles to
the triangulation as shown in Figure~\ref{fig:sperner} by the dotted lines: by adding
a sequence of triangles that have the original extremal vertex coloured 1, together
with two consecutive vertices on the 1-0 edge, we end up with a triangulation that
has only one 0/1-coloured edge on the exterior.

Construct a directed graph $G$ whose vertices are triangles of this extended
triangulation. Add a directed edge of $G$ between any two triangles that are adjacent
and separated by a 0/1-edge; the direction of the edge is so as to cross with 0 on its
left and 1 on its right. Consequently, there is a single edge coming into the triangulation
from the outside. It is simple to check that in $G$, trichromatic triangles correspond
exactly with degree-1 vertices, solutions to {\sc End of the line}. This completes
the reduction.
\end{proof}
\begin{figure}
\begin{center}
\ifcolourfigs
\includegraphics[width=4.0in]{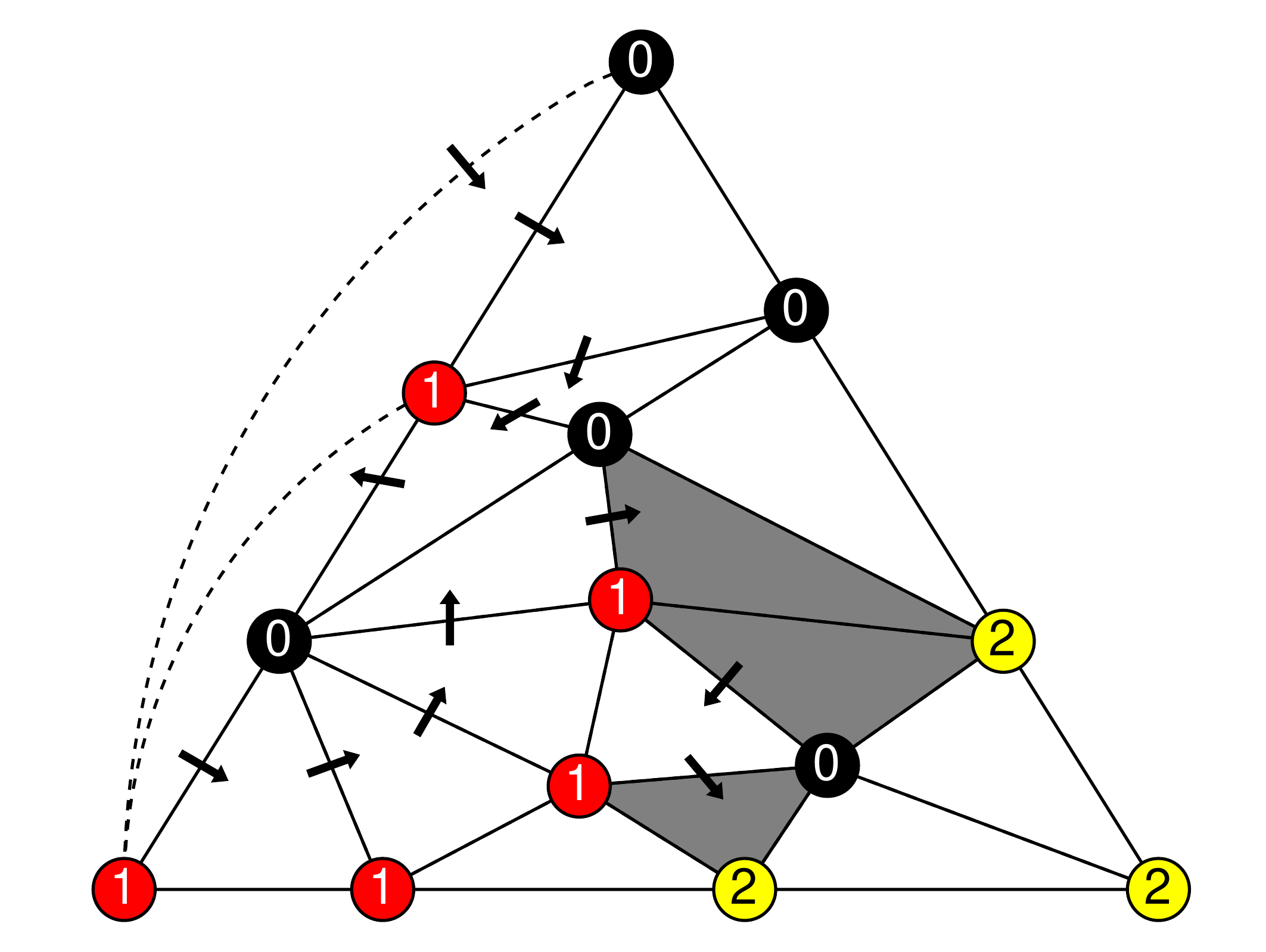}
\else
\includegraphics[width=4.0in]{2d-Sperner.pdf}
\fi
\end{center}
\caption{An example to illustrate Sperner's lemma.\newline
The solid lines show the triangle and the simplicisation (triangulation, since we
are in 2 dimensions.) The dashed lines add an additional sequence of topological
triangles along one side of the triangulation; as a result of those lines, there is
only one edge coloured 0-1 when viewed from the exterior.\newline
The arrows are the directed edges of the corresponding graph.\newline
The shaded triangles are trichromatic (End-of-line solutions in the corresponding graph).
}\label{fig:sperner}
\end{figure}
To define a challenging computational problem involved with searching for a panchromatic
Sperner simplex, we need to work with Sperner triangulations that are represented
so compactly that it becomes infeasible to just check every simplex. Hence, we
consider exponentially-large simplicial decompositions that satisfy the boundary
conditions required for Sperner's lemma.

Informally, the computational problem {\sc Sperner}~\cite{Pap} (with parameter $n$)
takes as input a Sperner triangulation that contains an number of vertices that is
exponential in $n$. The vertices ---and their colours--- cannot be explicitly listed,
since problem instances are supposed to be written down with a syntactic description
length that is polynomial in $n$. Instead, an instance is specified by a circuit $C$
that takes as input the coordinates of a vertex and outputs its colour. This means
that the vertices should lie on a regular grid, and $C$ takes as input a bit string that
represents the coordinates of a vertex.

Figure~\ref{fig:sperner-grid} shows one way to do this in 2 dimensions (assume that the
central grid has exponentially-many points). There is no restriction on how $C$
may colour the vertices labelled $*$; the other labels are fixed boundary
conditions to ensure that any trichromatic triangle lies within the grid. A 3-dimensional
version would just require some choice of simplicial decomposition of a cube, which
would then be applied to each small cube in the corresponding 3-d grid.

\begin{figure}
\begin{center}
\ifcolourfigs
\includegraphics[width=4.0in]{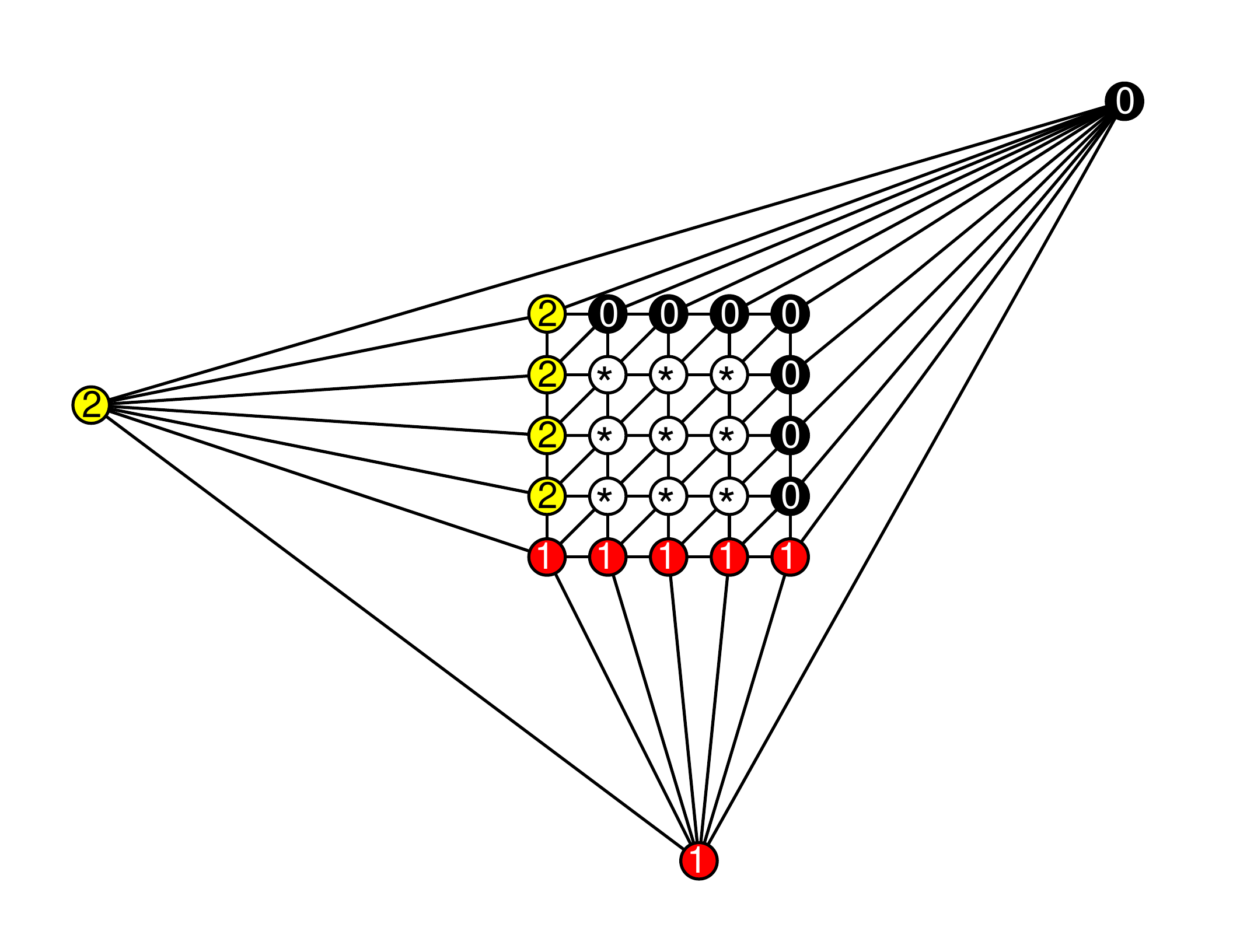}
\else
\includegraphics[width=4.0in]{Sperner-grid.pdf}
\fi
\end{center}
\caption{Embedding a triangulated grid inside an instance of Sperner's lemma
}\label{fig:sperner-grid}
\end{figure}

2-dimensional {\sc Sperner} is known to be \PPAD-complete~\cite{CD-icalp}; previously
it was known from~\cite{Pap} that 3-dimensional {\sc Sperner} is \PPAD-complete.
One interesting feature of the \PPAD-completeness results for Nash equilibrium computation,
is that normal-form games do not incorporate generic boolean circuits in an
obvious way. Contrast that with {\sc Sperner} as defined above, where
a problem instance incorporates a generic boolean circuit to determine the
colour of a vertex. We continue by explaining Nash equilibrium computation in more
detail, and how it is polynomial-time equivalent to {\sc Sperner}.

\section{Games and Nash Equilibria}\label{sec:ne}

A game $\G$ specifies a finite set of $\numplayers\geq 2$ players, where each player $\pindex$ has a
finite set $S_\pindex$ of (at least two) {\em actions} or {\em pure strategies}.
Let $S=S_1\times\ldots\times S_\numplayers$ be the set of {\em pure strategy profiles},
thus an element of $S$ represents a choice made by each player $\pindex$ of a single pure
strategy from his set $S_\pindex$. Finally, given any member of $S$,
$\G$ needs to specify real-valued {\em utilities} or {\em payoffs} to each
player that result from that pure strategy profile. Let $u^\pindex_s$ denote the payoff to
player $\pindex$ that results from $s\in S$.

Informally, a Nash equilibrium is a set of strategies ---one for each player--- where
no player has an ``incentive to deviate'' i.e. play some alternative strategy. After
giving some examples, we provide a precise definition, along with some notation. In
general, Nash equilibria do not exist for pure strategies; it is usually necessary to
allow the players to randomise over them, as shown in the following examples.
\begin{figure}[ht]
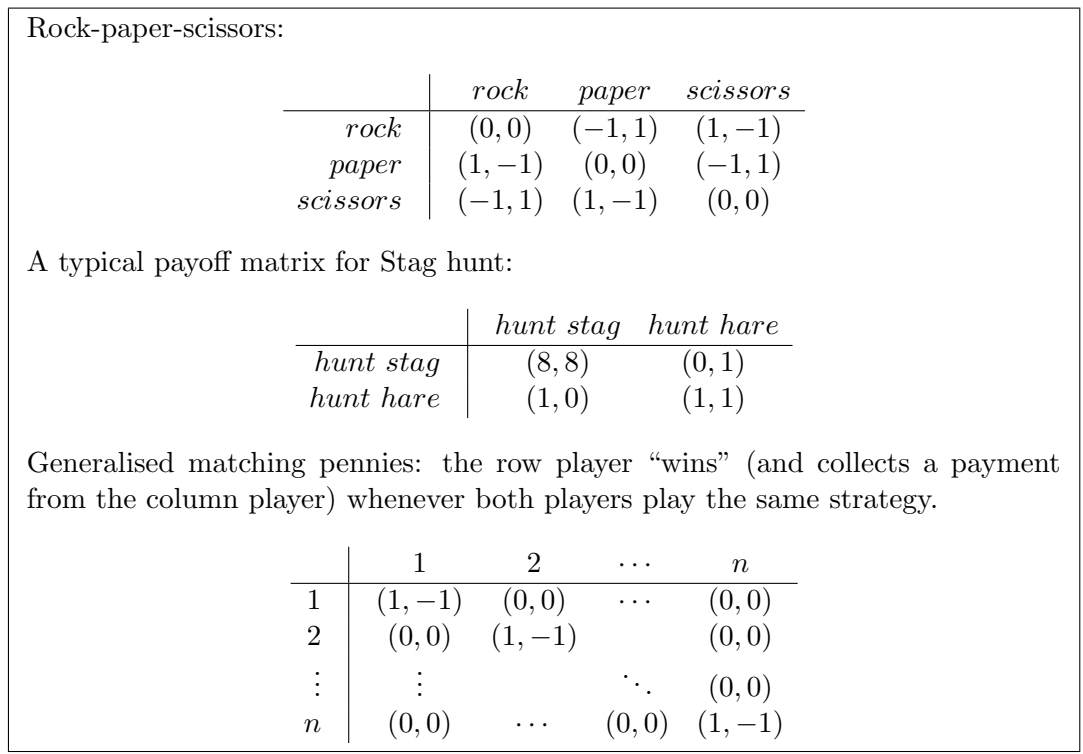

\begin{center}
\framebox{
\begin{minipage}{\boxwidth}
Rock-paper-scissors:
\[
\begin{array}{rcccc}
          & \vline & rock   & paper   & scissors \\ \hline
 rock     & \vline & (0,0)  & (-1,1)  & (1,-1)  \\
 paper    & \vline & (1,-1) & (0,0)   & (-1,1)  \\
 scissors & \vline & (-1,1) & (1,-1)  & (0,0)   \\
\end{array}
\]

A typical payoff matrix for Stag hunt:
\[
\begin{array}{rccc}
              & \vline & hunt~stag & hunt~hare     \\ \hline
    hunt~stag & \vline & (8,8)     & (0,1) \\
    hunt~hare & \vline & (1,0)     & (1,1) \\
\end{array}
\]

Generalised matching pennies: the row player ``wins'' (and collects a payment from the
column player) whenever both players play the same strategy.
\[
\begin{array}{rccccc}
            & \vline &   1    &   2    & \cdots & \numstrats \\ \hline
    1       & \vline & (1,-1) & (0,0)  & \cdots &   (0,0)    \\
    2       & \vline & (0,0)  & (1,-1) &        &   (0,0)    \\
 \vdots     & \vline & \vdots &        & \ddots &   (0,0)    \\
 \numstrats & \vline & (0,0)  & \cdots & (0,0)  &   (1,-1)   \\
\end{array}
\]
\end{minipage}
}\caption{Payoff matrices for example games. Each entry of a matrix contains two
numbers: the row player's payoff and the column player's payoff.} \label{fig:rps}
\end{center}
\end{figure}

\begin{example}
The traditional game of {\em rock-paper-scissors} fits in to the basic paradigm
considered here. The standard payoff matrix shown in Figure~\ref{fig:rps} awards one point for
winning and a penalty of one point for losing. The unique Nash equilibrium has both
players randomising uniformly over their strategies, for an expected payoff of zero.
\end{example}

\begin{example}
The {\em stag hunt} game~\cite{Rousseau} has two players, each of whom has two actions,
corresponding to hunting either a stag, or a hare. In order to catch a stag, both player must
cooperate (choose to hunt a stag), but either player can catch a hare on his own.
The benefit of hunting a stag is that the payoff (half share of a stag) is presumed to be
substantially larger than an entire hare.

A typical choice of payoffs (Figure~\ref{fig:rps}) to reflect this dilemma could award 8 to both
players when they both hunt a stag, 0 to a player who hunts a stag while the other player hunts
a hare, and 1 to a player who hunts a hare. Using these payoffs, there are 3 Nash equilibria:
one where both players hunt a stag, one where both players hunt a hare, and one where each player
hunts a stag with probability $\frac{1}{8}$ and a hare with probability $\frac{7}{8}$.
\end{example}

\begin{example}\label{ex:gmp}
{\em Generalised matching pennies} (Figure~\ref{fig:rps}) is a useful example to present here,
since it is used as an ingredient of the \PPAD-completeness proof for {\sc Nash}.

The game of {\em matching pennies} is a 2-player game, whose basic version has each player
having just two strategies, ``heads'' and ``tails''. The row player wins whenever both
players make the same choice, otherwise the column player wins. In the original version, a
win effectively means that the losing player pays one unit to the winning player; here we
modify payoffs so that a win for the column player just results in zero payoffs (he avoids
paying the row player). With that modification,  generalised matching pennies is the
extension to $\numstrats$ actions rather than just 2.

There is a unique Nash equilibrium in which both players use the uniform distribution over
their pure strategies, and it is easy to see that no other solution is possible.
\end{example}

\begin{paragraph}{Comments.}
The Stag Hunt example shows that there may be multiple equilibria, that some equilibria
have more social welfare than others, and that some are more ``plausible'' that others
(thus, it seems reasonable to expect one of the two pure-strategy equilibria to be played,
rather than the randomised one). The topic of {\em equilibrium selection} considers
formalisations of the notion of plausibility, and the problem of computing the relevant
equilibria. Of course, it is at least as hard to compute one of a subset of equilibria
as it is to compute an unrestricted one. For example, it is \NP-complete to
compute equilibria that guarantee one or both players a certain level of payoff, or that
satisfy various other properties that are efficiently checkable~\cite{CS, GZ}.

If we know the supports of a Nash equilibrium (the strategies played with non-zero
probability) then it would be straightforward to compute a Nash equilibrium, since it then
reduces to a linear programming problem. Thus (as noted by Papadimitriou in~\cite{NRTV})
the search for a Nash equilibrium is ---in the 2-player case--- essentially a combinatorial problem.
\end{paragraph}

\begin{paragraph}{Nash equilibrium and $\epsilon$-Nash equilibrium; definition and notation.}
A {\em mixed strategy} for player $\pindex$ is a distribution on $S_\pindex$, that is,
$|S_\pindex|$ nonnegative real numbers summing to 1. We will use $\msprob^\pindex_\aindex$ to
denote the probability that player $\pindex$ allocates to strategy $\aindex\in S_\pindex$.
If some or all players use
mixed strategies, this results in expected payoffs for the players. A {\em best response} by
a player is a strategy (possibly mixed) that maximises that player's expected payoff; the
assumption is that players do indeed play to maximise expected payoffs.

Call a set of $\numplayers$ mixed strategies $\msprob^\pindex_\aindex$ a {\em Nash equilibrium}
if, for each player $\pindex$, $\pindex$'s expected payoff
$\sum_{s\in S} u^\pindex_s \prod_{r=1}^\numplayers \msprob^r_{s_r}$ is maximized over all
mixed strategies of $\pindex$. That is, a Nash equilibrium is a set of
mixed strategies from which no player has an incentive to deviate.
Let $S_{-\pindex}=S_1\times\ldots\times S_{\pindex-1}\times S_{\pindex+1}\times\ldots\times S_\numplayers$,
the set of pure strategy profiles of players other than $i$.
For $s\in S_{-\pindex}$, let $\msprob_s = \prod_{r\not=i; r\in [k]} x^r_{s_r}$ and
$u^\pindex_{\aindex s}$ be the payoff to $\pindex$ when $\pindex$ plays $\aindex$ and
the other players play $s$.
It is well-known (see, e.g.,~\cite{OR}) that the following is an equivalent condition for a set
of mixed strategies to be a Nash equilibrium:
\begin{equation}
\forall\pindex,\aindex,\aindex' \sum_{s\in S_{-\pindex}} u^\pindex_{\aindex s}\msprob_s
> \sum_{s\in S_{-\pindex}} u^\pindex_{\aindex' s}\msprob_s
\Longrightarrow \msprob^\pindex_{\aindex'} = 0.
\end{equation}

Also, a set of mixed strategies is an $\epsilon$-Nash equilibrium for
some $\epsilon >0$ if the following holds:
\begin{equation}
\forall\pindex,\aindex,\aindex' \sum_{s\in S_{-\pindex}} u^\pindex_{\aindex s}\msprob_s
> \sum_{s\in S_{-\pindex}} u^\pindex_{\aindex' s}\msprob_s + \epsilon
\Longrightarrow \msprob^\pindex_{\aindex'} = 0.
\end{equation}
The celebrated theorem of Nash~\cite{Nash} states that every game has a Nash equilibrium.
\end{paragraph}

\begin{paragraph}{Comments.}
An $\epsilon$-Nash equilibrium is a weaker notion than a Nash equilibrium; a Nash equilibrium is an
$\epsilon$-Nash equilibrium for $\epsilon=0$, and generally an $\epsilon$-Nash equilibrium is
an $\epsilon'$-Nash equilibrium for any $\epsilon' >\epsilon$.
\end{paragraph}

\begin{paragraph}{Why focus on approximate equilibria?}
Approximate equilibria matter from the computational perspective, because for games of more than
2 players, a solution (Nash equilibrium) may be in irrational numbers, even when the utilities
$u^\pindex_s$ that specify a game, are themselves rational numbers~\cite{Nash}. The problem of
{\em computing} a Nash equilibrium ---as opposed to just knowing that one exists--- requires
us to specify a format or syntax in which to output the quantities that make up a solution
(i.e. the probabilities $\msprob^\pindex_s$).

In~\cite{cacmDGP} (an expository paper on the results of~\cite{DGP}) we noted an
analogy between equilibrium computation and the computation of a root of an odd-degree polynomial $f$
in a single variable. For both problems, there is a guarantee that a solution really exists, and
the guarantee is based on the nature of the problem rather than a promise that the given instance has
a solution. Furthermore, in each problem we have to deal with the issue of how to represent a
solution, since a solution need not necessarily be a rational number. And, in both cases a
natural approach is to switch to some notion of approximate solution: instead of searching for
$x$ with $f(x)=0$, search for $x$ with $|f(x)|<\epsilon$, which ensures that $x$ can be written
down in a standard syntax.
\end{paragraph}

\subsection{Some reductions among equilibrium problems}

By way of example, consider the well-known result (that predates the \PPAD-hardness
of {\sc Nash}; see~\cite{NRTV} Chapter 2 for further background) that symmetric 2-player games
are as hard to solve as general ones.

Given a $n\times n$ game $\G$, construct a symmetric $2n\times 2n$ game $\G'=f(\G)$, such
that given any Nash equilibrium of $\G'$ we can efficiently reconstruct a Nash
equilibrium of $\G$.
To begin, let us assume that all payoffs in $\G$ are positive --- the reason why this
assumption is fine, is that if there are any negative payoffs, then we can add a sufficiently
large constant to all payoffs and obtain a strategically equivalent version (i.e. one
that has the same Nash equilibria).

Now suppose we solve the $2n\times 2n$ game
$\G' = \left( \begin{array}{cc} 0 & \G \\  \G^T & 0 \end{array}\right)$, where we assume
that entries contain payoffs to both players, and the zeroes represent $n\times n$
matrices of payoffs of zero to both players.

Let $p$ and $q$ denote the probabilities that players 1 and 2 use their first $n$ actions, in
some given solution. If we label the rows and columns of the payoff matrix with the
probabilities assigned to them by the players, we have
\[
\begin{array}{cc}
 & {~~q ~~~ 1-q} \\
\begin{array}{r} {p} \\ {1-p} \end{array}
 &
\left( \begin{array}{cc} 0 & \G \\  \G^T & 0 \end{array} \right)
\end{array}
\]
If $p=q=1$, both players receive payoff 0, and both have an incentive to change their behaviour,
by the assumption that $\G$'s payoffs are all positive (and similarly if $p=q=0$).
So we have $p>0$ and $1-q>0$, or alternatively, $1-p>0$ and $q>0$.

Assume $p>0$ and $1-q>0$ (the analysis for the other case is similar).
Let $\{{p_1,...,p_\numstrats}\}$
be the probabilities used by player 1 for his first $n$ actions, $\{{q_1,\ldots q_n}\}$ the
probabilities for player 2's second $\numstrats$ actions.
\[
\begin{array}{cc}
 & {~~~q ~~~~ (q_1...q_n)} \\
\begin{array}{r} {(p_1,...p_n)} \\ {1-p} \end{array}
 &
\left( \begin{array}{cc} 0 & \G \\  \G^T & 0 \end{array} \right)
\end{array}
\]

Note that ${p_1+\ldots+p_n=p}$ and ${q_1+\ldots+q_n=1-q}$.
Then ${(p_1/p,\ldots,p_n/p)}$ and ${(q_1/(1-q),\ldots,q_n/(1-q))}$ are a Nash equilibrium
of $\G$. To see this, consider the
diagram; they should form a best response to each other for the top-right part.

This is an example of the kind of  reduction defined in Section~\ref{sec:reductions}.
In the context of games, this kind of reduction is called in~\cite{AKV} a {Nash homomorphism}.
Abbott et al.~\cite{AKV} reduce general 2-player games to win-lose 2-player games
(where payoffs are 0 or 1). Various other Nash homomorphisms have been
derived independently of the work relating {\sc Nash} to \PPAD. An
important one is Bubelis~\cite{Bubelis} that reduces (in a more
algebraic style) $k$-player games to 3-player games. The result also
highlights a key distinction between $k$-player games (for $k\geq 3$)
and 2-player games; in 2-player games the solutions are rational numbers
(provided that the payoff in the games are rational) while for 3 or more
players, the solutions may be irrational. The reduction
of~\cite{Bubelis} preserves key algebraic properties of the quantities
$\msprob^\pindex_\aindex$ in a solution, such as their degree (i.e the degree of the
lowest-degree polynomial with integer coefficients satisfied by
$\msprob^\pindex_\aindex$). Since we have noted that 2-player games have
solutions in rational numbers, this kind of reduction could not apply to
2-player games.

\subsection{The ``in \PPAD'' result}

A proof that  Nash equilibrium computation belongs to \PPAD\ necessarily incorporates a proof
of Nash's theorem itself (for approximate equilibrium), in that a reduction from
{\sc $\epsilon$-Nash} to {\sc End of the line} assures us that {\sc $\epsilon$-Nash}
is a total search problem, since it is clear that {\sc End of the line} is. To
derive Nash's theorem itself, that an exact equilibrium exists, the summary as in~\cite{Pap}
is as follows: Consider an infinite sequence of solutions to {\sc $\epsilon$-Nash} for
smaller and smaller $\epsilon$. Since the space of mixed-strategy profiles is compact,
these solutions have a limit point, which must be an exact solution.

\subsection{The Algebraic Properties of Nash Equilibria}

A reader who is interested on the combinatorial aspects of the topic can skip most of
this subsection; the main point to note is that since a Nash equilibrium is not necessarily
in rational numbers, this motivates the focus on approximate equilibria (i.e. $\epsilon$-Nash
equilibria), to ensure that there is a natural syntax in which to output a computed solution to a
game. Since any exact equilibrium is an approximate one, any hardness result for approximate
equilibria applies automatically to exact ones. The question of polynomial-time computation of
approximate equilibria was introduced in~\cite{GP} in order to finesse the irrationality
of exact solutions. However, algorithms for approximate equilibria go back earlier, notably
Scarf's algorithm~\cite{Scarf}.

For positive $\epsilon$, an $\epsilon$-Nash equilibrium need not be at all close to a true
Nash equilibrium. Etessami and Yannakakis~\cite{EY} show that even for exponentially small
$\epsilon$, an $\epsilon$-Nash equilibrium may be at variation distance 1 from any true
Nash equilibrium, for a 3-player game. Furthermore, they also show that computing an
approximate equilibrium that is within variation distance $\epsilon$ from a true one
(an {\em $\epsilon$-near} equilibrium) is square-root-sum hard. This refers to the following
well-known computational problem:
\begin{Def}
The {\em square root sum} problem takes as input two sets of positive integers, say
$\{x_1,\ldots,x_n\}$ and $\{y_1,\ldots,y_n\}$. The question is: is the sum of the square
roots of the first set, greater than the sum of square roots of the second?
\end{Def}
The problem is not known to belong to \NP, although neither is it known to be hard for any
well-known complexity class.
In terms of upper-bounding the complexity of this problem, and also for that matter,
computing an $\epsilon$-near equilibrium, we
may note that the criteria for a set of numbers $\msprob^\pindex_\aindex$ to be an
$\epsilon$-near equilibrium can be expressed in the existential theory of real arithmetic,
which places the problem in \PSPACE~\cite{Renegar}.

The fact that Nash equilibrium probabilities can be irrational numbers (while payoffs are
rational numbers) appears in an example in Nash's paper~\cite{Nash}. He describes a simple
poker-style game with 3 players and a unique (irrational) Nash equilibrium.
We noted earlier that any game with rational payoffs has Nash equilibria probabilities that
are algebraic numbers.
But even for a 3-player game, if $n$ is the number of actions, the solution may require algebraic
numbers of degree exponentially large in $n$; as noted in~\cite{GP}, the constructions introduced
there give us a flexible way to construct, from a given polynomial $p$, an arithmetic circuit
whose fixpoints are the roots of $p$, and~\cite{GP} shows how to construct a 4-player game
whose equilibria encode those fixpoints. The paper of Bubelis~\cite{Bubelis}
gave an ``algebraic'' reduction from any $\numplayers$-player game (for $\numplayers>3$) to
a corresponding 3-player game, in a way that preserves the algebraic properties of the solutions.
Consequently 3-player games have the same feature of 4-player games identified in~\cite{GP}.
The reduction of~\cite{Bubelis} can in addition be used to show that 3-player games are as
computationally hard to solve as $\numplayers$-players, but it does not highlight the
expressive power of games as arithmetic circuits. Recently, Feige and Talgam-Cohen~\cite{FT} give
a reduction from (approximate) $\numplayers$-player {\sc Nash} to 2-player {\sc Nash}, that does
not explicitly proceed via {\sc End of the line}.


\begin{paragraph}{Is Nash equilibrium computation ever harder than \PPAD?}

Daskalakis et al.~\cite{DFP} show that most standard classes of concisely-represented games
(such as graphical games and polymatrix games)
have equilibrium computation problems that reduce to {\sc 2-Nash} and so belong to \PPAD.
Schoenebeck and Vadhan~\cite{SV} study the question for ``circuit games'', a very general
concise representation of games where payoffs are given by a boolean circuit. In this context
they obtain hardness results, but this is not an \NP\ total search problem. The zero-sum version~\cite{FIKU}
is also hard.
\end{paragraph}

\section{Brouwer functions, and discrete Brouwer functions}\label{sec:brouwer}

Brouwer's fixpoint theorem states that every continuous function from a convex compact domain
to itself, must have a {\em fixpoint}, a point $x$ for which $f(x)=x$.

\begin{paragraph}{Proving Brouwer's fixpoint theorem using Sperner's lemma}
Suppose to begin with that the domain in question is the $d$-simplex $\Delta^d$. We
have continuous $f:\Delta^d\longrightarrow\Delta^d$ and we seek a point $x\in\Delta^d$
with $f(x)=x$. Suppose $f$ is evaluated on a set $S$ of ``sample points'' in $\Delta^d$,
and at each point $x\in S$, consider the value of $f(x)-x$. If the $d+1$ vertices of $\Delta^d$
are given distinct colours $\{0,1,\ldots,d\}$, we colour-code any $x\in S$ according
to the direction of $f(x)-x$: we give it the colour of a vertex which is at least as distant from
$f(x)$ as from $x$. Such a colouring respects the constraints of Sperner's lemma.
If $S$ is used as the vertices of
a simplicial decomposition, a panchromatic simplex becomes a plausible location for a fixpoint
($f$ is displacing co-located vertices in all different directions). One obtains a fixpoint from
the limit of increasingly fine simplicial decompositions.

The result can then be extended to other domains that are topologically equivalent to
simplices, such as cubes/cuboids, as used in the constructions we describe here.
\end{paragraph}

In defining an
associated computational total search problem ``find the point $x$'', we need to identify
a syntax or format in which to represent a class of continuous functions. The format
should ensure that $f$ can be computed in time polynomial in $f$'s description length
(so that solutions $x$ are easily checkable, and the search problem is in \NP). The
class of functions that we define are based on {\em discrete Brouwer functions},
defined in detail below in Section~\ref{sec:dbfs}.

Discrete Brouwer functions form the bridge between the discrete circuit problem {\sc End of
the line} and the continuous, numerical problem of computing a Nash equilibrium.

Consider the following obvious fact\footnote{pointed out in~\cite{CDT} (Section 5.1)} that if
we colour the integers $\{0,\ldots,n\}$ such that
\begin{itemize}
\item each integer is coloured either red or black, and
\item 0 is coloured red, and $n$ is coloured black
\end{itemize}
then there must be two consecutive numbers that have different colours.

Now suppose that $n$ is exponentially large, but for any number $x$ in the range $\{0,\ldots,n\}$
we have an efficient test to identify the colour of $x$. In that case, we still have an efficient
search for consecutive numbers that have opposite colours, namely we can use binary search. If
$\lfloor n/2\rfloor$ is coloured black, we would search for the pair of numbers in the
range $\{0,\ldots,\lfloor n/2\rfloor \}$,
otherwise we would search in the range $\{\lfloor n/2\rfloor,\ldots,n\}$, and so on.

Now consider a 2-dimensional version with pairs of numbers $(x,y)$,
$x,y\in\{0,\ldots,n\}$, where pairs of numbers are coloured as follows.
\begin{itemize}
\item if $x=0$, $(x,y)$ is coloured red,
\item if $y=0$ and $x>0$ then $(x,y)$ is coloured yellow,
\item if $x,y>0$ and either $x=n$ or $y=n$ (or both), then $(x,y)$ is coloured black.
\end{itemize}
Thus, we have fixed the colouring of the perimeter of the grid of points, but we impose no constraints
on how the interior should be coloured. We claim that there exists a square
of size 1 that has all 3 colours, but notice that binary search no longer finds
such a square efficiently\footnote{Hirsch et al.~\cite{HV} give lower bounds for algorithms
that search for fixed points based on function evaluation. In this setting, an algorithm that
searches for a fixed point of $f$ has ``black-box access'' to $f$; $f$ may be computed on
query points but the internal structure of $f$ is hidden. They contrast the one-dimension
case with the 2-dimensional case by showing that in 2 dimensions, the search for an approximate
fixed point of a Lipschitz continuous function in 2D, accurate to $p$ binary digits,
requires $\Omega(c^p)$ for some constant $c$, in contrast to 1D where bisection may be used.}.
The fact that such a square exists is due to Sperner's
lemma, whose proof envisages following a path that enters the $n\times n$ grid,
ending up at a trichromatic square; intuitively, this can be achieved by entering
at the point on the perimeter where red and yellow are adjacent, and walking around the
outside of the red region. Since there is only one exterior red/yellow interval,
we are guaranteed to reach a stage where the red region is no longer bounded by
yellow points, but by black points.

\begin{figure}
\begin{center}
\ifcolourfigs
\includegraphics[width=6.0in]{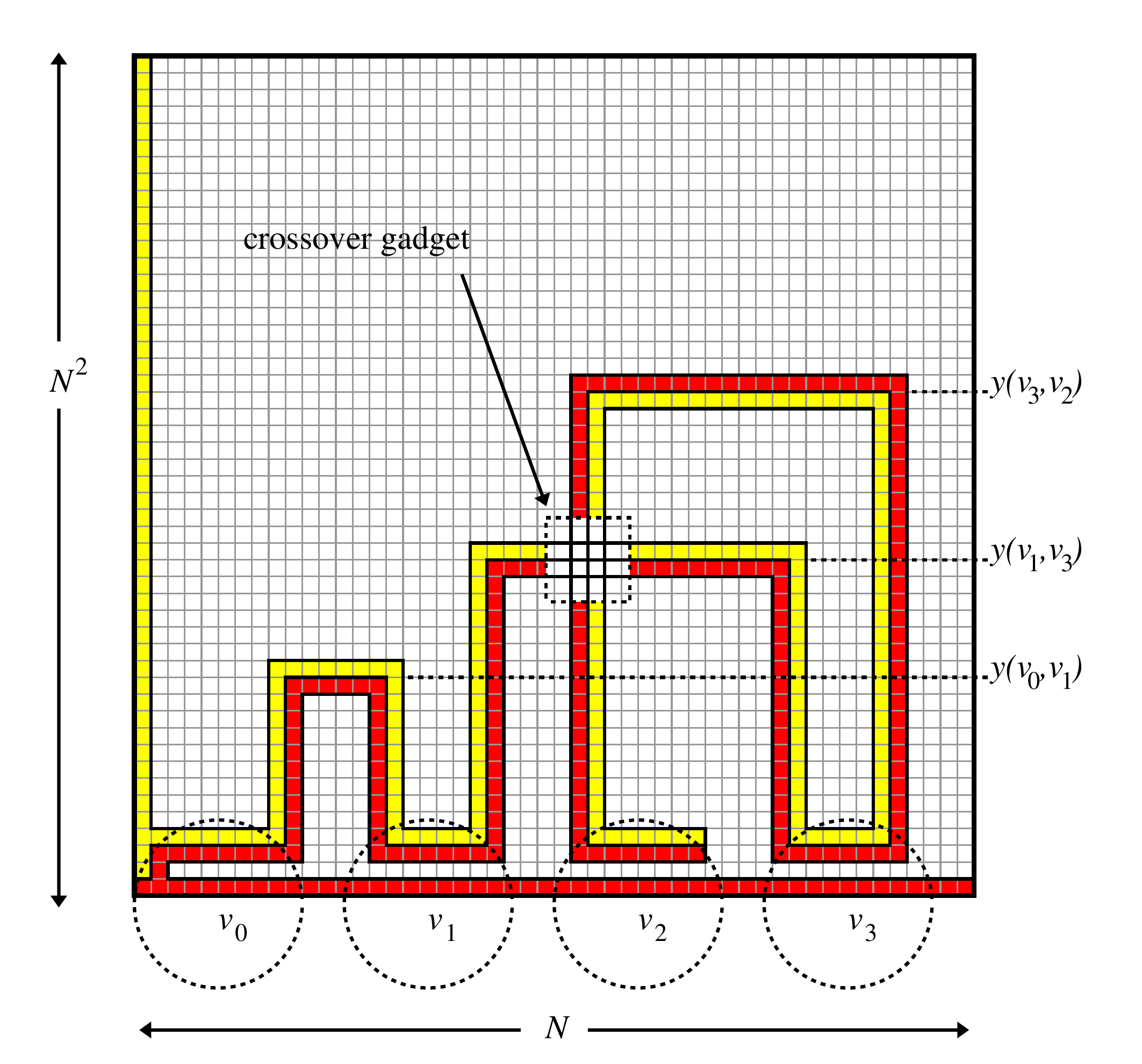}
\else
\includegraphics[width=6.0in]{eotl-dbf.pdf}
\fi
\end{center}
\caption{The figure shows how the line graph $(v_0,v_1), (v_1,v_3), (v_3,v_2)$ is encoded
as a discrete Brouwer function in two dimensions.\newline
The function colours the grid as a long thin red/yellow strip on a black background.
Each vertex $v_i$ is associated with a horizontal line segment close to the $x$-axis (circled
in the diagram). Each edge $(v_i,v_j)$ is associated with a ``bridge'' from the right-hand side
of $v_i$'s line segment to the left-hand side of $v_j$'s.\newline
Each bridge consists of three red/yellow line segments, going up, across and down. The
$y$-coordinate of the horizontal segment $y(i,j)$ should efficiently encode the values $i$ and $j$
(for example, $y(i,j)=i + N.j$ where $N$ is the number of $x$-values possible).\newline
Hence no two bridges can have overlapping horizontal sections. Unfortunately, a horizontal
section may need to cross a vertical section of a different bridge, as shown in the diagram.
Here a crossover gadget of~\cite{CD-icalp} (Figure~\ref{fig:dbf-crossover-gadget}) may be
used.
}\label{fig:eotl-dbf}
\end{figure}

\begin{figure}
\begin{center}
\ifcolourfigs
\includegraphics[width=3.0in]{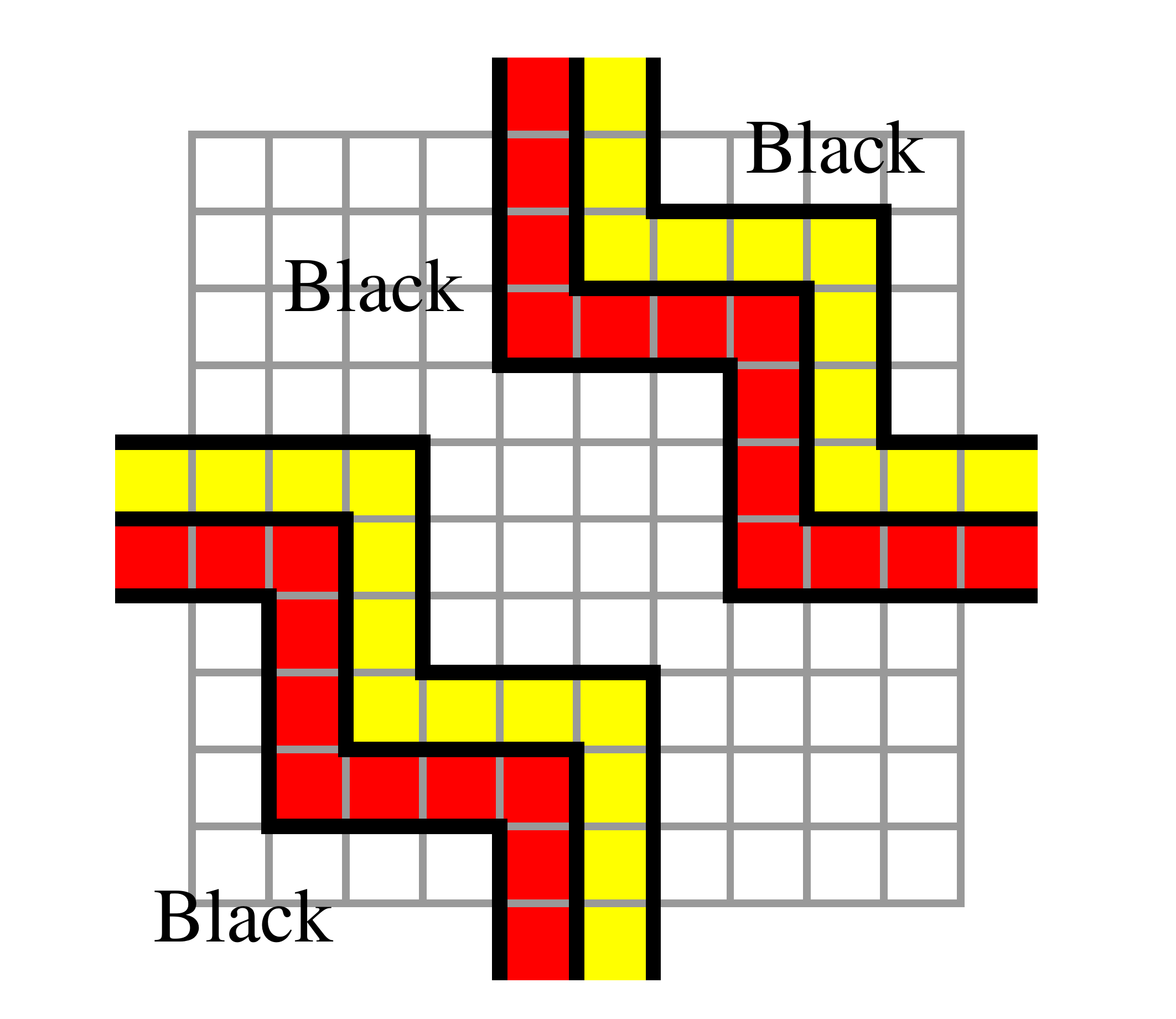}
\else
\includegraphics[width=3.0in]{dbf-crossover-gadget.pdf}
\fi
\end{center}
\caption{The crossover gadget referred to in Figure~\ref{fig:eotl-dbf}.\newline
To avoid all three colours meeting in the vicinity of the crossing of two
red/yellow path segments: the simple solution is to re-connect the paths in a way
that the topological structure of the red-yellow line no longer mimics the
structure of the {\sc End of the line} graph from which it was derived. The 3-dimensional
version, invented in~\cite{Pap} and refined in~\cite{DGP} does not require crossover gadgets.
}\label{fig:dbf-crossover-gadget}
\end{figure}

\subsection{Discrete Brouwer functions}\label{sec:dbfs}

We next give a detailed definition of discrete Brouwer functions (DBFs), together
with an associated total search problem (on an exponential-sized domain), that is
used in the reductions from {\sc End of the line} to the search for a
Nash equilibrium of a game. Figure~\ref{fig:eotl-dbf} shows how an {\sc End of the line}
graph is encoded by a DBF.

\begin{Def}
Partition the unit $d$-dimensional cube $\cube$ into ``cubelets''---
$2^{dn}$ axis aligned cubes of edge length $2^{-n}$.
A {\em Brouwer-mapping circuit} (in $d$ dimensions) takes as input a bit string of
length $dn$ that represents the coordinates of a cubelet, and outputs its colour, subject
to boundary constraints that generalise the 2-D case discussed above.
A {\em Discrete Brouwer function} (DBF) is the function computed by such a circuit,
and its syntactic complexity is the size of the circuit that represents it, which we restrict
to be polynomial in $n$.
\end{Def}

This gives rise to the total search problem of finding a vertex that belongs to cubelets of all the
different colours (a ``panchromatic cubelet'').
A superficial difference between the presentation of this concept in~\cite{DGP}
compared with~\cite{CDT} is that \cite{DGP} associates each cubelet with a colour
while \cite{CDT} associated each vertex with a colour, and the search is for a cubelet that
has vertices of all colours.

\section{From Discrete Brouwer functions to Games}

The reduction is broken down into two stages: first we take a circuit $C$ representing a DBF
and encode it as an arithmetic circuit $C'$ that computes a continuous function from
$\cube^d$ to $\cube^d$, where $\cube^d$ is the $d$-dimensional unit cube. Then we express
$C'$ as a game $\G$ in such a way that given any Nash equilibrium of $\G$, we can efficiently
extract the coordinates of a fixpoint of $C'$.

Figure~\ref{fig:dbf} gives the general idea of how to construct a continuous function $f$ from
a DBF, in such a way that fixpoints of the continuous function correspond to solutions of
the DBF. The idea is that each colour in the range of the DBF corresponds to a direction of
$f(x)-x$. For a colour that lies on the boundary of the discretised cube, we choose a direction
away from that boundary, so that $f$ avoids displacing points outside the cube. Indeed, this
is where we require the boundary constraints on the colouring of DBFs. Points at the centres
of the cubelets of the DBF will be displaced in the direction of that cubelet's colour, while
points on or near the boundary of two or more cubelets should get directions that interpolate
the individual colour-directions (which is necessary for continuity). A natural choice for
these colour-directions results in the following property: {\em For a proper subset of these
directions, their average cannot be zero. On the other hand, a weighted average of all the
directions can indeed be zero.}

\begin{figure}
\begin{center}
\ifcolourfigs
\includegraphics[width=5.0in]{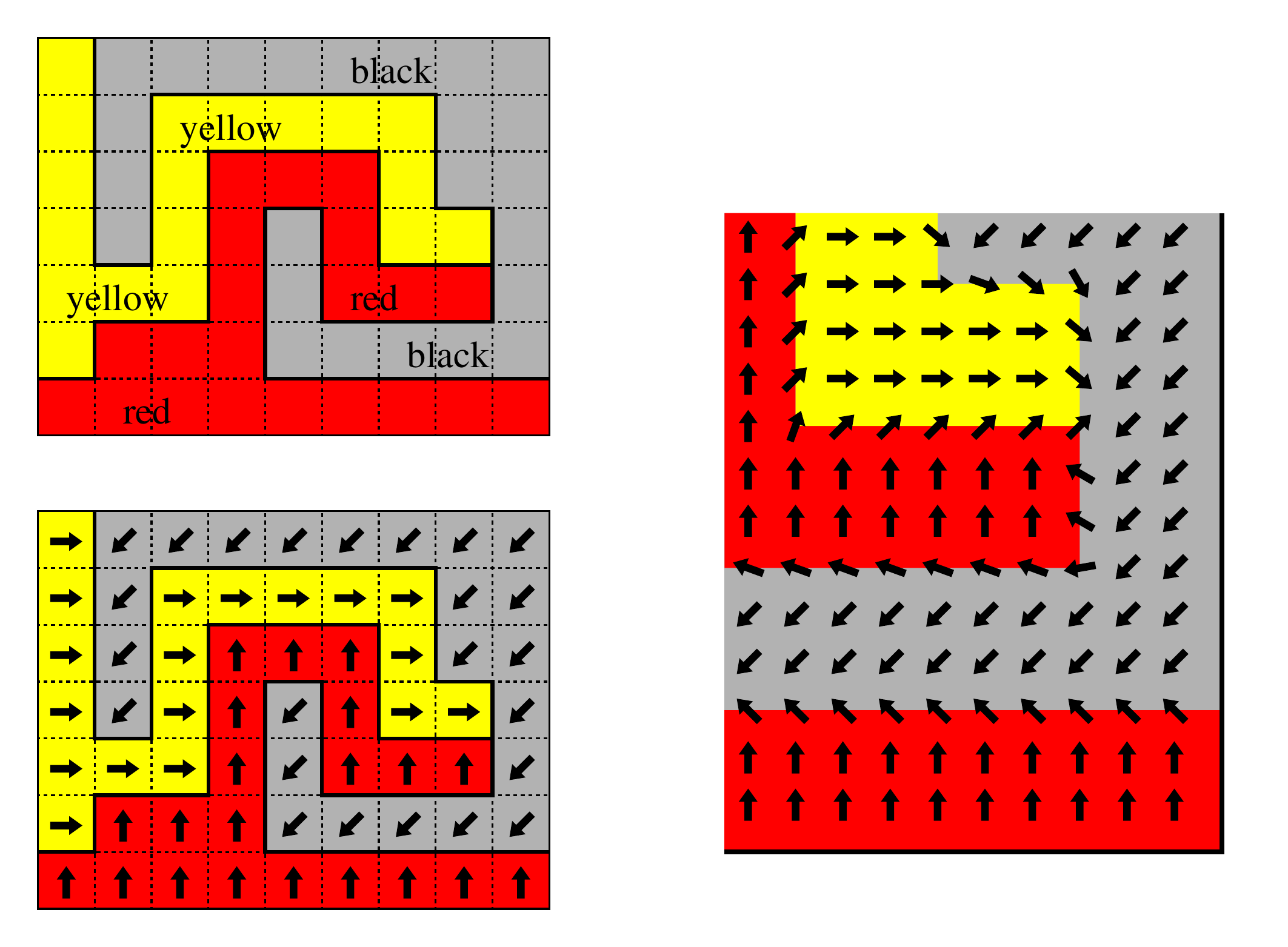}
\else
\includegraphics[width=5.0in]{2d-Brouwer.pdf}
\fi
\end{center}
\caption{The top left-hand diagram shows a simple discrete Brouwer function in
2 dimensions, and the arrows in the bottom left-hand diagram show the direction
of the corresponding continuous Brouwer function $f$ within each small square.\newline
The right-hand diagram is a magnified view of the bottom right-hand section of
the continuous Brouwer function, showing how the direction of the function
(i.e the direction of $f(x)-x$) can be interpolated on the boundaries of the
small squares.\newline
It is always possible to smoothly interpolate between 2 different directions without introducing
fixed points of $f$. But, in the vicinity of the point where all 3 colours meet, there will
necessarily be a point $x$ where $f(x)=x$.
}\label{fig:dbf}
\end{figure}

\begin{paragraph}{Arithmetic circuits}
The papers~\cite{CDT,DGP} show how to construct arithmetic circuits that compute these kinds of
continuous functions. An arithmetic circuit consists of a directed graph each of whose nodes
belongs to one of a number of distinct types, such as addition and multiplication. Each node
will compute a real number, and the type of a node dictates how that number is obtained. For
example, an addition node should have 2 incoming arcs, and it will compute the sum of the
values located at the adjacent nodes for those two arcs; it may also have a number of outgoing
arcs that allow its value to feed in to other nodes.

A significant difference between the circuits of~\cite{DGP} and those of~\cite{CDT} is that the
latter do not allow nodes that compute the {\em product of two computed quantities} (located
at other nodes with incoming edges). They do allow the multiplication of a computed
quantity by a constant. The result of this is
\begin{itemize}
\item The functions that can be computed are more constrained; however they are still
expressive enough to simulate the behaviour of DBF circuits,
\item two-player normal-form games can simulate the circuits of~\cite{CDT}, while three-player
games are required to simulate circuits that can multiply a pair of computed quantities.
\end{itemize}
The more general circuits introduced in~\cite{DGP} correspond to the algebraic complexity
class FIXP of Etessami and Yannakakis~\cite{EY}; the ones without these product nodes correspond
to \PPAD, the linear version of FIXP, called LINEAR-FIXP in~\cite{EY}.
\end{paragraph}

\begin{paragraph}{Interpolating the directions}
For a point $x$ that is {\em not} on a cubelet boundary, one can design a polynomial-sized
arithmetic circuit that compares the values of its coordinates against various numbers,
eventually determining which cubelet it belongs to. Its colour assigned by the original
DBF can be computed using a further component of the arithmetic circuit, that simulates the
DBF. A further component of the circuit can translate the colour to a direction-vector,
which gets added to $x$ to yield $f(x)$.

For $x$ on or close to a cubelet boundary, both papers~\cite{CDT,DGP} perform the interpolation
as follows. $f$ is computed not just at $x$, but also at a small cluster of
points in the vicinity of $x$. Then the average value of $f(x')$ is computed, for all points
$x'$ in this cluster. \cite{CDT} show how this can be done in non-constant dimension, which
allows a stronger hardness-of-approximation result to be obtained; see below. For constant
dimension this clustering trick can be avoided; \cite{GPS} shows the interpolation can
be done more directly, and without the multiplication nodes discussed above.
\end{paragraph}

\begin{paragraph}{Encoding approximate solutions, and snake embeddings}
{\em Snake embeddings}, invented by Chen et al.~\cite{CDT} allow hardness results for
$\epsilon$-approximate solutions where $\epsilon$ is inverse polynomial in $n$.
A snake embedding maps a discrete Brouwer function ${\cal F}$ to a lower-dimension
DBF ${\cal F}'$, in such a way that a solution to ${\cal F}'$ efficiently encodes a solution
to ${\cal F}$, and the number of cubelets along the edges decreases by a constant factor.
Applying these repeatedly, we obtain a DBF in $\Theta(n)$ dimensions where the cubelets
have $O(1)$ edge length. The coordinates of points in the cube can be perturbed by relatively large
amounts without leaving a cubelet.
\end{paragraph}

\subsection{Graphical Games}

Graphical games were introduced in~\cite{KLS,LKS} as a succinct means of representing certain
games having many players. In a graphical game, each player has an associated vertex of an
underlying graph $G$. The payoff to a player at vertex $v$ is assumed to depend only on
the behaviour of himself and his neighbours. Consequently, for a low-degree graph, the
payoffs can be described much more concisely than would be allowed by normal form.

Graphical games were used by Daskalakis et al.~\cite{DGP} as an intermediate stage between
discrete Brouwer functions, and normal-form games. That is, it was shown in~\cite{DGP}
how a discrete Brouwer function can be converted into a graphical game $\GG$ such that any
Nash equilibrium of $\GG$ encodes a solution to the discrete Brouwer function.
Essentially, the underlying graph of $\GG$ has the same structure as the arithmitic circuit
derived from the DBF, and the probabilities in the Nash equilibrium correspond to
the quantities computed in the circuit.
Chen et al~\cite{CDT} reduce more directly from discrete Brouwer functions to 2-player games.

\subsection{From graphical to normal-form games}

Recall the game of {\em generalised matching pennies} (GMP) from Example~\ref{ex:gmp}.
The construction used in~\cite{CDT,DGP} starts with a ``prototype'' zero-sum game consisting of
a version of GMP where all strategies have been duplicated, as shown in
the following payoff matrix $R$ for the row player ($M$ is some large positive quantity).
\[
R =
\left(
\begin{array}{ccccccc}
M & M & 0 & 0 & \cdots & 0 & 0 \\
M & M & 0 & 0 & \cdots & 0 & 0 \\
0 & 0 & M & M & \cdots & 0 & 0 \\
0 & 0 & M & M & \cdots & 0 & 0 \\
\vdots & \vdots & \vdots & \vdots & \ddots & \vdots & \vdots \\
0 & 0 & 0 & 0 & \cdots & M & M \\
0 & 0 & 0 & 0 & \cdots & M & M \\
\end{array}
\right)
\]
Having noted that GMP has a unique equilibrium in which both players randomise uniformly,
we next note that in the above game, each player will randomise uniformly over pairs of
duplicated strategies. So, each player allocates probability $\frac{1}{\numstrats}$ to
each such pair. Within each pair, that allocation of $\frac{1}{\numstrats}$ may be split
arbitrarily; $\G$ imposes no constraints on how to divide it.

The idea next, is to add certain quantities to the payoffs (that are relatively small in
comparison with $M$) so that a player's choice of how to split between a pair of ``duplicated''
strategies may affect the opponents' choice for other pairs of his strategies. It can be
shown that in Nash equilibria of the resulting game,
\begin{itemize}
\item the probability allocated to any strategy-pair is not exactly $\frac{1}{\numstrats}$ but
is fairly close, since the value of $M$ dominates the other payoffs that are introduced,
\item the probabilities $p$ and $p'$ allocated to the members of a strategy-pair are
used to represent the number $p/(p+p')\in [0,1]$,
\item the numbers thus represented can be made to affect each others' values in the same way
that players in a graphical game, having just 2 pure strategies $\{0,1\}$, affect each others'
probabilities of playing 1.
\end{itemize}

\section{Easy and hard classes of games}

The \PPAD-completeness results for normal-form games have led to a line of
research addressing the very natural question of what types of games admit
polynomial-time algorithms, and which ones are also \PPAD-hard. In particular, with regard
to \PPAD-hardness, the general aim is to obtain hardness results for games that are
more and more syntactically restricted; the first results, for 4 players~\cite{DGP-stoc},
then 3 players~\cite{CD-eccc,DP-eccc}, and then 2 players~\cite{CD-focs} can be seen
as the initial chapter in this narrative. In this section, our focus is on results
on the frontier of \PPAD-completeness and membership of \Poly. We do not consider
games where solutions can be shown to exist by means of a potential function argument
(known to be equivalent to {\em congestion games}~\cite{MS,Rosenthal}).

While normal-form games are the most natural ones to consider from a theoretical
perspective, there are many alternative ways to specify a game. It is, after all,
unnatural to write down a description of a game in normal form; only for very
small games is this feasible.

\subsection{Hard equilibrium computation problems}

So, one natural direction is to look for \PPAD-hardness results for more restricted
types of games than the ones currently known to be \PPAD-hard.

\begin{paragraph}{Restricted 2-player games} We currently
know that 2-player games are hard to solve even when they are sparse~\cite{CDT,CDT-wine};
the problem {\sc Sparse bimatrix} denotes the problem of computing a Nash equilibrium
for 2-player games having a constant bound on the number of non-zero entries in each
row and column. (For example, Theorem 10.1 of~\cite{CDT} show that 2-player games remain
\PPAD-hard when rows and columns have up to 10 non-zero entries.)
2-player games are also \PPAD-complete to solve
when they have 0/1 payoffs~\cite{AKV}\footnote{The paper of Abbott et al.~\cite{AKV}
appeared before the first \PPAD-hardness results for games, and reduced the search
for Nash equilibria of general 2-player games to the search for Nash equilibria
in $0/1$-payoff 2-player games (i.e. win-lose games).}.
\end{paragraph}

\begin{paragraph}{Restricted graphical games}
It is shown in~\cite{EGG} that degree-2 graphical games are solvable in polynomial
time, but \PPAD-complete for graphs with constant pathwidth (it is an open question
precisely what pathwidth is required to make the problem \PPAD-complete).
\end{paragraph}

\begin{paragraph}{Ranking games}
Brandt et al.~\cite{BFHS} study ranking games from a computational-complexity
viewpoint. Ranking games are proposed as model of various real-world competitions;
the outcome of the players' behaviour is a ranking of its
participants, thus is maps any pure-strategy profile to a ranking of the players.
Let $u^\pindex_r$ be the payoff to player $\pindex$ that
results from being ranked $r$-th, where we assume $u^\pindex_r\geq u^\pindex_{r+1}$
(players prefer to be ranked first to being ranked second, and so on.) 

In the 2-player case, any ranking game is strategically equivalent to a zero-sum game.
For more than 2 players, one can apply the observation of~\cite{vNM} that any $k$-player
game is essentially equivalent to a $k+1$-player zero-sum game, where the additional
player is given payoffs that set the total payoff to zero, but takes no part in the
game, in that his actions do not affect the other players' payoffs. Using this, it is
easy to reduce 2-player win-lose games to 3-player ranking games.
\end{paragraph}

\begin{paragraph}{Polymatrix games}
A polymatrix game is a multiplayer game in which any player's payoff is the sum of payoffs
he obtains from bimatrix games with the other players. A special case of interest is where
there is a limit on the number of pure strategies per player.
In general, these polymatrix games are \PPAD-complete~\cite{DFP}, which also
follows from~\cite{DGP} --- specifically
Section 6 of~\cite{DGP} presents a \PPAD-hardness proof for {\sc 2-Nash} by using graphical
game gadgets that have the feature that the payoff to a vertex can be expressed as the sum
of payoffs of 2-player games with his neighbours. (This extension to {\sc 2-Nash} applies the
observation that it is unnecessary to have a player who computes the product of his neighbours;
such a vertex has payoffs that cannot decompose as the sum of bimatrix games.)
\end{paragraph}

\subsection{Polynomial-time equilibrium computation problems}

In reviewing some of the computational positive results (types of games that have
polynomial-time algorithms) we focus on games where the search problem is guaranteed
to be total due to being in \PPAD, as opposed to {\em potential games} for example,
where equilibria can be shown to exist due to a potential function argument.

\begin{paragraph}{Restricted win-lose games}
The papers~\cite{ABOV,CLR} identify polynomial-time algorithms for subclasses of the win-lose games
shown to be hard by~\cite{AKV}. A win-lose game has an associated bipartite graph whose vertices
are the pure strategies of the two players. We add an edge from $s$ to $s'$
if when the players play $s$ and $s'$, the player of $s$ obtains a payoff of
1. Addario-Berry et al.~\cite{ABOV} study the special case when this graph is
planar, and obtain an algorithm based on a search in the graph
for certain combinatorial structures within the graph, that runs in polynomial
time for planar graphs. Codenotti et al.~\cite{CLR} study win-lose games in
which the number of winning entries in each row or column is at most 2.
Their approach also involves searching for certain structures within a
corresponding digraph.
\end{paragraph}

\begin{paragraph}{Approximate equilibria}
The main question is: for what values of $\epsilon>0$ can one compute $\epsilon$-Nash
equilibria in polynomial time? The question has mainly been considered in the 2-player case.
Note that~\cite{CDT} established that there is not a fully polynomial-time approximation
scheme for computing $\epsilon$-Nash equilibria. However, it is also known that for
any constant $\epsilon>0$, the problem is subexponential~\cite{LMM}.

There has not been much progress in the past couple of years on reducing the value of
$\epsilon$ that is obtainable; the papers cited in~\cite{DGP} represent the state of the art
in this respect. Perhaps the simplest non-trivial algorithm for computing approximate
Nash equilibria is the following, due to Daskalakis et al.~\cite{DMP}.
Figure~\ref{fig:approxalgo} shows a generalisation of the algorithm of~\cite{DMP}, to
$\numplayers$ players, obtained independently in~\cite{BGR, HRS}.
These papers also obtain the lower bounds for solutions having constant support. The
approximation guarantee is $1-\frac{1}{\numplayers}$; it is noted in~\cite{BGR} that one can
do slightly better by using a more sophisticated 2-player algorithm at the midpoint of the
procedure.
\end{paragraph}

\begin{figure}[ht]
\begin{center}
\framebox{
\begin{minipage}{\boxwidth}
\begin{enumerate}
\item For $\pindex = 1,2,\ldots,\numplayers-1$
\begin{enumerate}
\item Player $\pindex$ allocates probability $1-\frac{1}{k+1-i}$ to some arbitrary strategy
\end{enumerate}
\item For $\pindex = k,k-1,\ldots,1$
\begin{enumerate}
\item Player $\pindex$ allocates his remaining probability to a best response to the
strategy combination played so far.
\end{enumerate}
\end{enumerate}
\end{minipage}
}\caption{Approximation algorithm for $\numplayers$-player {\sc Nash}\label{fig:approxalgo}}
\end{center}
\end{figure}

\begin{paragraph}{Ranking games}
A restriction of the ranking games mentioned above, to those having
``competitiveness-based strategies'' were recently proposed in~\cite{GGKV} as a subclass of
ranking games that seems to have better computational properties, while still being
able to capture features of many real-world competitions for rank. In these games, a player's
strategies may be ordered in a sequence that reflects how ``competitive'' they are. If we let
$\{a_1,\ldots,a_n\}$ be the actions available to a player, then each $a_j$ has two associated
quantities, a cost $c_j$ and a return $r_j$, where return is a monotonically increasing function
of cost. Given any pure-strategy profile, players are ranked on the returns of their actions
and awarded prizes from that ranking. The payoff to a player is the value of the prize he wins,
minus the cost of his action. This results in a trade-off between saving on cost, versus spending
more with the aim of a larger prize. Notice that, in contrast to unrestricted ranking
games~\cite{BFHS}, this restriction allows games with many players to be written down concisely.
\end{paragraph}

\begin{paragraph}{Anonymous games}
Anonymous games represent a useful way to concisely describe certain games having a large
number of players --- in an anonymous game, each player has the same set of pure strategies
$\{a_1,\ldots,a_k\}$, and the payoff to player $i$ for playing $a_j$ depends only on $i$ and
the total number of players to play $a_j$ (but not the identities of those players; hence the
phrase ``anonymous game''). Daskalakis and Papadimitriou~\cite{DP08} give a polynomial-time approximation
scheme for these games, but note that it is an open problem whether an exact equilibrium may be
computed in polynomial time.
\end{paragraph}

\begin{paragraph}{Polymatrix games}
Daskalakis and Papadimitriou~\cite{DP09} show that polymatrix games may be solved exactly when the
bimatrix games between pairs of players are zero-sum. This result is applied in~\cite{GGKV} to
a subclass of the games studied there, specifically ranking games where
prize values are a {\em linearly} decreasing function of rank placement. 
\end{paragraph}

\section{The Complexity of Path-following Algorithms}

In this section we report on recent progress~\cite{GPS} that shows an even closer analogy
between equilibrium computation and the problem {\sc End of the line}. Namely, that the
solutions found by certain well-known ``path-following'' algorithms have the same
computational complexity as the search for the solution to {\sc End of the line} that is
obtained by following the line. Specifically, both are \PSPACE-complete.

Consider the algorithm for {\sc End of the line} that works by simply ``following the line''
from the given starting-point $0^n$ until an endpoint is reached.
Clearly this takes exponential time in the worst case,
but in fact we can say something stronger. Let {\sc Oeotl} denote the problem ``other end
of this line'', which requires as output, the endpoint reached by following directed edges
from the given starting-point. Papadimitriou observed in~\cite{Pap} that the following holds:
\begin{theorem}\cite{Pap}
{\sc Oeotl} is \PSPACE-complete.
\end{theorem}
Notice that a solution to {\sc Oeotl} is apparently no longer in \NP, since while we can
check that it solves {\sc End of the line}, there is no obvious way to efficiently check
that it is the correct end-of-line, i.e. the one connected to the given starting-point.
\begin{proof}(sketch)
We reduce from the problem of computing the final configuration of a polynomial space bounded
Turing machine (TM).

The proof uses the fact that polynomial-space-bounded Turing machines can be simulated by
polynomial-space-bounded {\em reversible Turing machines}~\cite{CP}. By a {\em configuration}
of a TM computation we mean a complete description of an intermediate state of computation,
including the contents of the tape, and the state and location of the TM on the tape.
Reversible TMs have the property that given an intermediate configuration,
one can readily obtain not just the subsequent configuration, but also the previous
one; this is done by memorising a carefully-chosen subset of previous configurations that the
machine uses in the course of a computation~\cite{CP}.

From there, it is straightforward to simulate a reversible Turing machine using the two
circuits $S$ and $P$ of an instance of {\sc Oeotl}. Each vertex $v$ of an $(S,P)$-graph
encodes a configuration of a linear-space-bounded TM; $S$ computes the subsequent
configuration, and $P$ computes the previous configuration.
\end{proof}

In the context of searching for Nash equilibria, a number of algorithms have been proposed
that work by following paths is a graph $G$ associated with game $\G$, where $G$ is
derived from $\G$ via a reduction to {\sc End of the line}. Examples include
Scarf's algorithm~\cite{Scarf} (for approximate solutions of $\numplayers$-player games)
and the Lemke-Howson algorithm~\cite{LH} (an exact algorithm for 2-player games).
These algorithms were proposed as being in practice more computationally efficient
than a brute-force approach. Related algorithms include
the {\em linear tracing procedure} and homotopy
methods~\cite{Harsanyi,Herings,HvdE,HP,HP2} discussed below: in these papers the
motivation is {\em equilibrium selection} --- in cases where multiple equilibria exist, we
seek a criterion for identifying a ``plausible'' one.

\begin{paragraph}{Homotopy methods for Game Theory}
In topology, a homotopy refers to a continuous deformation from some geometrical object
to another one. Suppose we want to solve game $\G$. Let $\G_0$ denote another game
obtained by changing the numerical utilities of $\G$ in such a way that there is some
``obvious'' equilibrium; for example, we could change the payoffs so that each player obtains
1 for playing his first strategy, and 0 for any other strategy, regardless of the behaviour
of the other player(s). Then, we can continuously deform $\G_0$ to get back to $\G$.
For $t\in[0,1]$ let $\G_t$ be a game whose payoffs are weighted averages of those in $\G_0$
with those in $\G$: we write this as $\G_t = (1-t)\G_0 + t\G$. A homotopy method for solving
$\G$ involves keeping track of the equilibria of $\G_t$, which should move continuously as $t$
increases from $0$ to $1$. There should be a continuous path from the known equilibrium of
$\G_0$ to some equilibrium of $\G$ (an application of Browder's fixed point theorem~\cite{Browder}.)
The catch is, that in following this path it may be necessary for $t$ to go down as well as up.
\end{paragraph}

\begin{figure}[ht]
\begin{center}
\framebox{
\begin{minipage}{\boxwidth}
The Lemke-Howson algorithm is usually presented as a {\bf discrete path-following algorithm},
in which a 2-player game has an associated degree-2 graph $G$ as follows. Vertices of $G$ are mixed
strategy profiles that are characterised by the subset of the players' strategies that are
either played with probability zero, or are best responses to the other player's mixed strategy.
To be a vertex of $G$, a mixed-strategy profile should satisfy one of the following:
\begin{itemize}
\item All pure strategies are labelled, or
\item All but one pure strategies are labelled, and one pure strategy is labelled twice,
due to being both a best response, and for being played with probability zero.
\end{itemize}
In the first case, we either have a Nash equilibrium, or else we have assigned probability
zero to all strategies. In the second case, it can be shown that there are two ``pivot''
operations that we may perform, that remove one of the duplicate labels of the
doubly-labelled strategy, and add a label to some alternative strategy; see von Stengel~\cite{BvS}
for details. These operations correspond to following edges on a degree-2 graph.

\medskip
Viewed as a {\bf homotopy method} (see Herings and Peeters~\cite{HP2} Section 4.1) the algorithm
works as follows. For a $\numstrats\times\numstrats$ game $\G$ with players 1 and 2, choose
any $\pindex\in\{1,2\}$ and $\aindex\in[\numstrats]$ and give player $\pindex$ a large bonus
payoff for using strategy $a^\pindex_\aindex$; the bonus should be large enough to make
$a^\pindex_\aindex$ a dominating strategy for player $\pindex$. Hence, there is a unique
Nash equilibrium consisting of pure-strategy $a^\pindex_\aindex$ together with the pure best
response to $a^\pindex_\aindex$. Now, we reduce that bonus to zero, and the equilibrium
should change continuously; however, the requirement that we keep track of a continuously
changing equilibrium means that in general, the bonus cannot reduce monotonically to zero;
it may have to go up as well as down.

\medskip
The choice of $a^\pindex_\aindex$ at the start of the homotopy corresponds to the initial
``dropped label'' in the discrete path-following description.

\end{minipage}
}\caption{Two views of the Lemke-Howson algorithm\label{fig:lh}}
\end{center}
\end{figure}

The main theorem of~\cite{GPS} is that
\begin{theorem}\cite{GPS}
It is \PSPACE-complete to compute any of the equilibria that are obtained by the
Lemke-Howson algorithm.
\end{theorem}

Savani and von Stengel~\cite{SvS} established earlier that the Lemke-Howson does indeed
take exponentially many steps in the worst case; this new result, then, says that moreover
there are no ``short cuts'' to the solution obtained by Lemke-Howson, subject only to the
very weak assumption that \PSPACE-hard problems require exponential time in the worst case.
(The \PPAD-hardness of {\sc Nash} already established that these solutions are hard to find
subject to the hardness of \PPAD, but the hardness of \PPAD\ is a stronger assumption.)

\begin{figure}
\begin{center}
\ifcolourfigs
\includegraphics[width=5.0in]{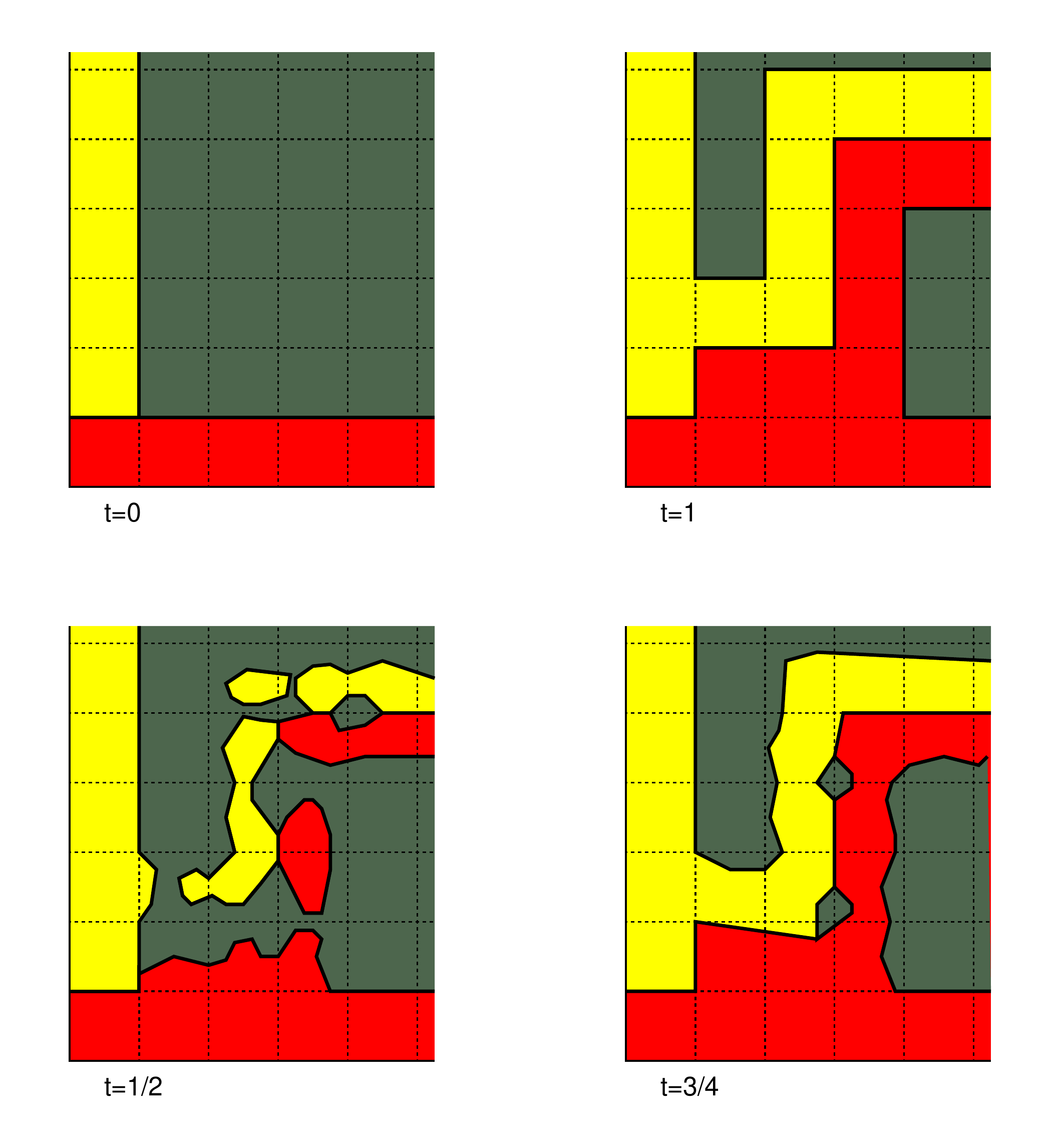}
\else
\includegraphics[width=5.0in]{homotopy-figure.pdf}
\fi
\end{center}
\caption{Homotopy for Brouwer function in the plane
}\label{fig:homotopy}
\end{figure}

\begin{proof}(the main ideas)
Recall the role of discrete Brouwer functions in the reductions from {\sc End of the line}
to {\sc Nash}, as discussed in Section~\ref{sec:brouwer}. Suppose $\G=\G_1$ is associated
with a discrete Brouwer function having the kind of structure indicated in
Figure~\ref{fig:eotl-dbf}; the bottom left-hand corner looks roughly like the top
right-hand ``$t=1$'' diagram in Figure~\ref{fig:homotopy}. Suppose that $\G_0$ is
associated with a DBF in which the line-encoding structure has been stripped out:
the top left ``$t=0$'' diagram in Figure~\ref{fig:homotopy}.

Consider the linear homotopy between the corresponding continuous Brouwer functions
$\F_0$ and $\F_1$; $\F_t(x)=(1-t)\F_0(x)+\F_1(x)$. $\F_t$ colours the square according
to the direction of $\F_t(x)-x$. Figure~\ref{fig:homotopy} (the bottom half) shows
how these colourings may evolve for intermediate values of $t$. For intermediate values
of $t$, in regions where $\F_0=\F_1$ we have $\F_t(x)=\F_0(x)=\F_1(x)$ for all $t$.
This means that fixpoints of $\F_t$ cannot be located in the region where
$\F_1$ gets the colour 0 (or black). So, a continuous path of fixpoints necessarily
follows the line of non-black regions of $\F_1$, and necessarily ends up at the
end of that particular line. The crossover gadget (Figure~\ref{fig:dbf-crossover-gadget})
would break this correspondence with {\sc Oeotl}; we fix that by moving to 3-dimensions,
where the gadget is not needed.
\end{proof}

\section{From Games to Markets}

The \PPAD-completeness results for variants of {\sc Nash} has led, more recently, to
\PPAD-completeness results for the problems of computing certain types
of market equilibria. We continue by giving the general idea and
intuition, then we proceed to the formal definitions. For more details on the
background to this topic, see Chapters 5,6 of~\cite{NRTV}; here we aim to focus on giving a
brief overview of more recent \PPAD-completeness results.

Suppose we have a set $G$ of {\em goods} and a set $T$ of {\em traders}. Each
trader $i$ in $T$ has a utility function $f^i$ that maps bundles of goods to non-negative
real values. Assume the goods are divisible, so $f^i$ is a mapping from vectors of non-negative
numbers (the quantity of each good in a bundle) to non-negative real numbers.
A standard assumption is that $f^i$ should be continuous, and non-decreasing when restricted to an
individual good (it doesn't hurt to receive a bigger quantity), and concave (there is non-increasing
marginal utility). 
Now, suppose that each good $j$ gets a (non-negative real) unit price $p_j$. A trader with
some given budget can then identify (from $f^i$) the optimal bundle of goods that his budget
will purchase (where optimal bundles need not be unique). In general, the aim is to identify
prices for goods such that each trader can exchange an initial allocation for an optimal one
that has the same total price, so that the total quantity of each good is conserved.
Under some fairly mild conditions, it can be shown that such
prices always exist. Intuitively, if all traders try to exchange their initial
allocations for ones that are optimal, and there is too much demand for good $j$ as a result,
then we can fix that problem by raising the price of $j$.

\begin{Def}
In an {\em Arrow-Debreu market}~\cite{AD}, each trader has an initial {\em endowment} consisting of
a bundle of goods. Suppose that the prices are set such that the following can happen:
each trader exchanges his endowment for an optimal bundle having the same total value, and
furthermore, the total quantity of each good is conserved. In that case, we say that these
prices allow the market to {\em clear}. It is shown in~\cite{AD} that there always exist
prices that allow the market to clear.

A {\em Fisher market} is a special case of an Arrow-Debreu market, in which traders are initially
endowed with quantities of money, and there are fixed quantities of goods for sale; in this
case the prices should ensure that if each trader buys a bundle of goods that is optimal
for his budget, then all goods are sold. This can be seen to be a special case of
Arrow-Debreu market, by regarding money as one of the goods in $G$. 
\end{Def}

The existence proof of market-clearing prices~\cite{AD} works by expressing a
market satisfying the relevant
constraints in terms of an {\em abstract economy}, a generalisation of the standard notion of a
game, and applying Nash's theorem. So, indirectly the proof uses Brouwer's fixed point theorem.
This indicates that \PPAD\ is the relevant complexity class for studying the hardness of
computing a price equilibrium.
Results about the computational complexity of finding market-clearing prices are in terms
of the types of utility functions of the traders.

Interesting classes of utility functions include the following (where $(x_1,\ldots,x_m)$
denotes a vector of quantities of $m$ goods):
\begin{enumerate}
\item Additively separable functions, where a trader's function $f^i(x_1,\ldots,x_m)$ is of
the form $f^i(x_1,\ldots,x_m)=\sum_{j=1}^m f^i_j(x_j)$. We still need to specify the
structure of the functions $f^i_j$ in order to have a well-defined problem.
\item piecewise linear: in conjunction with additively separable, this would require that each function
$f^i_j(x_j)$ be piecewise linear, and the syntactic complexity of such a function would be the
total number of pieces of the functions $f^i_j$. More generally, without the additively separable property,
$f^i$ could take the form
$f^i(x_1,\ldots,x_m)=\min_{a\in A} f^i_a(x_1,\ldots,x_m)$ where $A$ indexes a finite set of linear
functions, and $f^i_a(x_1,\ldots,x_m)=\sum_{j=1}^m\lambda^a_jx_j$ for non-negative
coefficients $\lambda^a_j$.
\item Leontiev economies: 
In a {\em Leontiev economy} we have $f^i(x_1,\ldots,x_m)=\min_{j\in[m]} \lambda^i_jx_j$, a special
case of the piecewise linear functions above; note that these functions are not however
additively separable.
\end{enumerate}

\begin{paragraph}{Polynomial-time algorithms}
Linear utility functions
take the form $f^i(x_1,\ldots,x_m)=\sum_{j=1}^m \lambda^i_jx_j$ --- they
are both additively separable and piecewise linear, but not Leontiev; they were initially
considered in~\cite{Gale}, and are known to be solvable in polynomial time~\cite{DPSV,Jain}
in the case of Fisher markets. For the Arrow-Debreu case, Devanur and Vazirani~\cite{DV03}
give a strongly polynomial-time approximation scheme but leave open the problem of
finding an exact one in polynomial time.
\end{paragraph}

\begin{paragraph}{PPAD-hardness results}
The first such result applied to Leontiev economies, for which there is a reduction
from 2-player games~\cite{CSVY}.
Chen et al.~\cite{CDDT} show that it is \PPAD-complete to compute an Arrow-Debreu market equilibrium
for the case of additively separable, piecewise linear and concave utility functions.
Chen and Teng~\cite{CT} show that it is \PPAD-hard to compute Fisher equilibrium prices,
from utility functions that are additively separable and piecewise linear concave; this
is done by reduction from {\sc Sparse bimatrix}~\cite{CDT-wine}.
Vazirani and Yannakakis~\cite{VY} show that finding an equilibrium in Fisher markets
is \PPAD-hard in the case of additively-separable, piecewise-linear, concave utility
functions. On the positive side, they show that with these utility functions, market
equilibria can however be written down with rational numbers. 
\end{paragraph}







\myaddress

\end{document}